\documentclass[journal]{IEEEtran}
\IEEEoverridecommandlockouts
\usepackage[utf8]{inputenc}
\usepackage{amsmath,epsfig} % easier math formulae, align, subequation
\usepackage{amssymb,amsthm}

\usepackage{url}
\usepackage{epstopdf}
\usepackage{multirow}
\usepackage{setspace}
\usepackage[noadjust]{cite}
\usepackage{xspace}
\usepackage{tabularx}
\usepackage{changepage}
\usepackage{color}
\usepackage{graphicx,bm}
\usepackage{color,cite}
\usepackage{times}
\usepackage{enumerate,type1cm}
\usepackage{amsfonts,relsize}
\usepackage{relsize}
\usepackage{fancybox}
\usepackage[english]{babel}
\usepackage{fancybox}
\usepackage{bbm}
\usepackage{tikz}
\usetikzlibrary{fit,positioning}
\usetikzlibrary{bayesnet}
\usepackage{orcidlink}
\usepackage{caption}
\def\BibTeX{{\rm B\kern-.05em{\sc i\kern-.025em b}\kern-.08em
		T\kern-.1667em\lower.7ex\hbox{E}\kern-.125emX}}

\newtheorem{thm}{\textbf{Theorem}}

\newtheorem{lem}{\textbf{Lemma}}

\long\def\symbolfootnote[#1]#2{\begingroup%
	\def\thefootnote{\fnsymbol{footnote}}\footnote[#1]{#2}\endgroup}

\newcommand{\beq}{\begin{equation}}
\newcommand{\eeq}{\end{equation}}
\newcommand{\beqa}{\begin{eqnarray}}
\newcommand{\eeqa}{\end{eqnarray}}
\usepackage{float}
\usepackage{colortbl}

\tikzset{
	startstop/.style={
		rectangle, 
		rounded corners,
		minimum width=3cm, 
		minimum height=0.5cm,
		align=center, 
		draw=black, 
	},
	process/.style={
		rectangle, 
		minimum width=3cm, 
		minimum height=0.5cm, 
		align=center, 
		draw=black, 
	},
	decision/.style={
		rectangle, 
		minimum width=3cm, 
		minimum height=0.5cm, align=center, 
		draw=black, 
	},
	arrow/.style={thick,->,>=stealth},
	dec/.style={
		ellipse, 
		align=center, 
		draw=black, 
	},
}


\makeatletter
\def\BState{\State\hskip-\ALG@thistlm}
\makeatother

	{\end{adjustwidth}}

\usepackage{array}

\makeatletter
\newcommand{\removelatexerror}{\let\@latex@error\@gobble}
\makeatother

\allowdisplaybreaks

\usepackage{cite}
\usepackage{amsmath,amssymb,amsfonts,amsthm}
\usepackage{algorithmic}
\usepackage{graphicx}
\usepackage{textcomp}
\usepackage{xcolor}
\usepackage{graphicx}
 \usepackage{marginnote}
 \usepackage{fancyhdr}
 \usepackage{epstopdf}
\usepackage{anyfontsize}
\usepackage{makeidx}
\usepackage[utf8]{inputenc}
\usepackage{multirow}
\usepackage{multicol}
\usepackage[none]{hyphenat}
\usepackage{lastpage}
\usepackage[T1]{fontenc}
\usepackage{url} %for proper url entries
\usepackage{epstopdf}
\usepackage{booktabs}
\usepackage{multicol}
\usepackage{fancyhdr}
\usepackage{indentfirst}
\usepackage{color}
\usepackage[section]{placeins}
\usepackage{float}
\usepackage[font=small]{caption}
\usepackage{subcaption}
\usepackage{amsmath, amsthm, amssymb}
\usepackage{bm,bbm}
\usepackage{commath}
\usepackage[ruled,vlined,linesnumbered]{algorithm2e}
\usepackage{stfloats}
\graphicspath{{./Figures/}}
\newcommand{\RNum}[1]{\uppercase\expandafter{\romannumeral #1\relax}}
\def\BibTeX{{\rm B\kern-.05em{\sc i\kern-.025em b}\kern-.08em
    T\kern-.1667em\lower.7ex\hbox{E}\kern-.125emX}}

\allowdisplaybreaks

\title{User Activity Detection for Irregular Repetition Slotted Aloha based MMTC}
\author{\IEEEauthorblockN{Chirag Ramesh Srivatsa, \IEEEmembership{Graduate Student Member, IEEE,} \orcidlink{0000-0002-3732-4733} and Chandra R. Murthy, \IEEEmembership{Senior Member, IEEE} \orcidlink{0000-0003-4901-9434}} %\vspace{-0.8cm}
\thanks{Chirag Ramesh Srivatsa is with the Robert Bosch Centre for Cyber-Physical Systems, Indian Institute of Science, Bangalore, India. Chandra R.  Murthy is with the Dept. of Electrical Communication Engineering, Indian Institute of Science, Bangalore, India (e-mail: \{chiragramesh, cmurthy\}@iisc.ac.in).  
} 
}

\begin{document}

\maketitle
\begin{abstract}
Irregular repetition slotted aloha (IRSA) is a grant-free random access protocol for massive machine-type communications, in which users transmit replicas of their packet in randomly selected resource blocks within a frame.
In this paper, we first develop a novel Bayesian user activity detection (UAD) algorithm for IRSA, which exploits both the sparsity in user activity as well as the underlying structure of IRSA transmissions.
Next, we derive the Cram\'{e}r-Rao bound (CRB) on the mean squared error in channel estimation. 
We empirically show that the channel estimates obtained as a by-product of the proposed UAD algorithm achieves the CRB.
Then, we analyze the signal to interference plus noise ratio achieved by the users, accounting for UAD, channel estimation errors, and pilot contamination.
Finally, we illustrate the impact of these non-idealities on the throughput of IRSA via Monte Carlo simulations. 
For example,  in a system with $1500$ users and $10\%$ of the users being active per frame, a pilot length of as low as $20$ symbols is sufficient for accurate user activity detection. 
In contrast, using classical compressed sensing approaches for UAD would require a pilot length of about $346$ symbols.
Our results reveal crucial insights into dependence of UAD errors and throughput on parameters such as the length of the pilot sequence, the number of antennas at the BS, the number of users, and the signal to noise~ratio.
\end{abstract}

\begin{IEEEkeywords}
Irregular repetition slotted aloha, grant-free random access, massive machine-type communications, user activity detection, channel estimation.
\end{IEEEkeywords}

%\vspace{-0.2cm}
\section{Introduction} \label{sec_intro}
Massive machine-type communications (MMTC) is expected to serve around a million devices per square kilometer~\cite{ref_shariatmadari_commag_2015}.
Typical MMTC devices transmit short packets to a central base station (BS), and are sporadically active~\cite{ref_xia_comstd_2019}.
To facilitate efficient random access for such MMTC applications, distributed grant-free random access (GFRA) protocols need to be used, as they can serve a large number of users without incurring a large signaling overhead~\cite{ref_liu_spmag_2018}.
Since only a subset of users are active in any frame in MMTC~\cite{ref_xia_comstd_2019},  it is essential for the BS to detect the set of users that are active,  before proceeding to perform channel estimation and data decoding.
This process is termed \emph{user activity detection (UAD).} 
Furthermore, without UAD, the BS would waste valuable resources attempting to decode a large number of users that have not transmitted any packets, i.e., users that are inactive.
Knowing the subset of active users not only saves computational resources by helping the BS decide which users it needs to decode,  it is also important for channel estimation, as will be seen in the sequel.
Errors arising from the UAD process, namely, false positives and false negatives, deteriorate the channel estimates computed at the BS, which  in turn affects the data decoding. 
Hence, it is crucial to account for these errors while analyzing the performance of GFRA protocols.

%\vspace{-0.2cm}
\subsection{Motivation}
Irregular repetition slotted aloha (IRSA) is a popular GFRA protocol in which users transmit replicas of their packets in multiple resource blocks (RBs) in a frame \cite{ref_liva_toc_2011}.
Each RB can accommodate a whole data packet.
In IRSA, each user samples their repetition factor $d$ from a predefined distribution independently of other users and then transmits replicas of its packet in $d$ RBs chosen uniformly at random from the set of all RBs in the frame~\cite{ref_liva_toc_2011}.
The indices of the RBs in which the users transmit their packet replicas define the \emph{access pattern matrix (APM),} which we assume is known at the BS. 
This assumption is explained in Sec.~\ref{sec_sys_model}.
Due to the structure of the APM, applying existing UAD algorithms to IRSA can lead to suboptimal performance.
In particular, it is necessary to combine the information available in each RB in a principled manner, to accurately detect the active users.

Typically, UAD and channel estimation is performed by the BS using pilots transmitted by the users in their packet headers.
If the users employ mutually orthogonal pilots, there is no pilot contamination, making UAD and channel estimation simple.
However, the length of orthogonal pilots scales linearly with the total number of users, and hence the pilot overhead quickly overshadows the data payload size as the number of users gets large~\cite{ref_liu_tsp_2018}. 
Thus, non-orthogonal pilots are used, and the resulting pilot contamination leads to both UAD errors and channel estimation errors. 
These effects must be accounted for while analyzing the performance of IRSA.
The main goal of this paper is to understand the effect of system parameters such as pilot length, SNR, and the number of antennas at the BS on the performance of IRSA, accounting for UAD errors, channel estimation errors, and pilot contamination.

%\vspace{-0.2cm}
\subsection{Related Works}
Early works in IRSA considered the collision model in which only \emph{singleton} RBs are decodable.
Singleton RBs are RBs in which only a single user has transmitted, and since there are no collisions in such RBs, users can be decoded with high probability.
The decoding proceeds in iterations, and occurs via inter-RB successive interference cancellation (SIC), which refers to the removal of a packet replica from all other RBs where the same packet was transmitted.
The decoding stops when there is no singleton RB available.
The throughput of IRSA under the collision model is at most one packet per RB~\cite{ref_liva_toc_2011}, which is achievable asymptotically with the number of RBs and users if the soliton distribution is used to generate the repetition factors~\cite{ref_narayanan_istc_2012}.

In the case where the BS is equipped with multiple antennas, multiple users could potentially be decoded in an RB \cite{ref_bjornson_mimo_2017}, and thus singleton RBs are not necessary for decoding. 
Any user with a sufficiently high instantaneous signal to interference plus noise ratio (SINR) can first be decoded, and the contribution of that user can be removed from the same RB.
This process, termed intra-RB SIC, refers to the removal of interference of a packet replica from the same RB within which it was decoded.
After the user with the highest SINR is decoded, other users could potentially be decoded as well.
By performing both intra-RB and inter-RB SIC, the packet replicas of different users are removed from all RBs wherein the same user has transmitted a packet.
This model, which we use in this paper, is termed as the \emph{SINR threshold} model, and it yields a higher throughput than the conventional singleton decoding model.
IRSA has been studied with the SINR threshold model under scalar Rayleigh fading channels with perfect channel state information (CSI)~\cite{ref_clazzer_icc_2017} and pure path loss channels~\cite{ref_khaleghi_pimrc_2017}.
Coded slotted aloha (CSA), which is a variant of IRSA,  was recently analyzed with imperfect SIC~\cite{ref_dumas_tvt_2021}.
The authors in~\cite{ref_valentini_globecom_2021} studied CSA with an acknowledgement mechanism between frames.
IRSA was analyzed with an SIC limit, i.e.,  a limit on the maximum number of packets that can be recovered in each RB, in~\cite{ref_shieh_tcom_2022}.
The average age of information in IRSA in MMTC has also been studied~\cite{ref_saha_jsac_2021}.
We have previously \cite{ref_srivatsa_spawc_2019, ref_srivatsa_chest_tsp_2021} analyzed the IRSA protocol accounting for channel estimation and pilot contamination, with perfect UAD. 
In contrast, this paper focuses on UAD in IRSA, and analyzes the impact of UAD errors on the throughput.

To the best of our knowledge, the problem of UAD in IRSA has not yet been considered in the literature.  
Further, none of the existing works study the performance of IRSA accounting for UAD errors, path loss, MIMO fading, pilot contamination, and channel estimation errors.
An initial study into estimating the number of active users in IRSA was conducted in \cite{ref_sun_pimrc_2018}, which does not identify the subset of active users.
UAD has been studied for massive random access outside the context of IRSA \cite{ref_fengler_tit_2021,ref_haghighatshoar_isit_2018}.
The activity matrix to be estimated  has jointly-sparse columns, i.e., columns that have the same sparse support~\cite{ref_guo_tcom_2022}.
Typical UAD solutions involve compressed sensing-based solutions~\cite{ref_senel_tcom_2018} or a maximum aposteriori probability (MAP) detection~\cite{ref_jeong_tvt_2018}.
The sparse Bayesian learning (SBL) framework has been employed to perform UAD in MMTC~\cite{ref_rajoriya_tcom_2022}.
Faster SBL algorithms for UAD in MMTC have also been developed~\cite{ref_zhang_tvt_2021}.
Other low complexity algorithms for UAD include approximate message passing~\cite{ref_zou_splet_2020} and orthogonal matching pursuit~\cite{ref_zhang_wicomlet_2022}.
These approaches, however, cannot be used in IRSA due to the structure imposed by the APM.
A na\"{i}ve approach would be to perform UAD on an RB-by-RB basis and declare users inactive if they are found to be inactive in all the RBs.
As we will show, this approach is inefficient and results in large error rates, especially when non-orthogonal pilots are used.

%\vspace{-0.2cm}
\subsection{Contributions}
This paper proposes a novel UAD algorithm for IRSA, and analyzes the throughput of IRSA, accounting for UAD and channel estimation errors.
Our main contributions are:
\begin{enumerate}
\item We develop a novel Bayesian algorithm to detect the set of active users in IRSA in Sec. \ref{sec_uad}. 
UAD in IRSA is a joint-sparse signal recovery problem with a measurement model with an important twist: different and unknown subsets of the row indices of the joint-sparse matrix participate in different measurements.
Our algorithm is an enhancement to the multiple sparse Bayesian learning (MSBL) algorithm~\cite{ref_wipf_tsp_2007} to cater to this scenario.\footnote{Our UAD algorithm can be applied to other variants of IRSA such as CSA since it entails only a minor change in the structure of the APM.}
\item We derive the channel estimates at the BS for users in all RBs in IRSA, acquired via non-orthogonal pilots, in Sec. \ref{sec_ch_est}. We also derive the Cram\'{e}r-Rao bound (CRB) on the mean squared error (MSE) of the channels estimated by our proposed UAD algorithm. We show that a genie-aided minimum MSE (MMSE) estimator (that has knowledge of the second-order statistics and the user activities) achieves the CRB. 
Later, we also empirically show that the MSE of the channel estimates output by the proposed UAD algorithm meets the CRB.
\item Next, we analyze the SINR achieved by all the users in all RBs in Sec.~\ref{sec_sinr}, accounting for UAD errors, channel estimation errors, and pilot contamination. 
The SINR expression allows us to determine the throughput of IRSA, accounting for the effect of UAD errors.  
\end{enumerate}

Our numerical experiments in Sec.~\ref{sec_results} show that there is at least a 4-fold reduction in the number of pilot symbols required to achieve a similar UAD performance as that of existing approaches.
The loss in performance due to UAD errors can be recuperated by judiciously choosing the system parameters such as pilot length, number of antennas, and SNR.
In essence, it is vital to account for both UAD and channel estimation when analyzing the throughput of IRSA.

\textit{Notation:} The symbols $a$, $ \mathbf{a}$, $\mathbf{A}$, $ [\mathbf{A}]_{i,:}$, $ [\mathbf{A}]_{:,j}$, $ \mathbf{0}_N$, $ \mathbf{1}_N,$ and $\mathbf{I}_N $ denote a scalar, a vector, a matrix, the $i$th row of $\mathbf{A}$, the $j$th column of $\mathbf{A}$, all-zero vector of length $N$, all ones vector of length $N$, and an identity matrix of size $N \times N$, respectively. 
$[\mathbf{a}]_{\mathcal{S}}$ and $[\mathbf{A}]_{:,\mathcal{S}}$ denote the elements of $\mathbf{a}$ and the columns of $\mathbf{A}$ indexed by the set $\mathcal{S}$ respectively. 
$\text{diag}(\mathbf{a})$ is a diagonal matrix with diagonal entries given by $\mathbf{a}$, whereas $\text{blkdiag}(\mathbf{A},\mathbf{B})$ is a block diagonal matrix with $\mathbf{A}$ and $\mathbf{B}$ as the diagonal blocks.
$\mathbf{A} \otimes \mathbf{B}$ is the Kronecker product of $\mathbf{A}$ and $\mathbf{B}$, and $\mathbf{A} \succeq \mathbf{B} $ denotes that $\mathbf{A} - \mathbf{B}$ is positive semi-definite.
$[N]$ denotes the set $\{1,2,\ldots,N\}$.
$|\cdot|$, $\|\cdot\|$, $\|\cdot\|_F$, $ [\cdot]^T $, $ [\cdot]^H $, $\mathop{{}\mathbb{E}}[\cdot]$, and $ \mathop{{}\mathbb{E}_\mathbf{a}}\left[ \cdot \right] $ denote the magnitude, $\ell_2 $ norm, frobenius norm, transpose, conjugate transpose, unconditional expectation, and the expectation conditioned on $ \mathbf{a}$, respectively. 
The superscript ${\tt{p}}$ is used as a descriptive superscript in association with a symbol that is related to the \emph{pilots}.
All the other superscripts (or subscripts) that have not been defined as above are indices.

%\vspace{-0.2cm}
\section{System Model} \label{sec_sys_model}
An IRSA system is considered with $M$ single-antenna users communicating with a BS equipped with $N$ antennas.
The users are assumed to be spread randomly within a cell, with the BS located at the cell center.
These users communicate with the BS over frames consisting of $T$ RBs. 
The RBs can be slots, subcarriers or both. 
In every frame, a small subset of the $M$ users, called \emph{active} users, attempt to deliver a packet each to the BS.
In a given frame, the \emph{activity coefficient} of the $m$th user is denoted by $a_m$, where $a_m = 1$ if the $m$th user is active, and $a_m = 0$ otherwise. 
Note that $a_m$ can change from one frame to the next, and the subset of active users (and hence $a_m$) is unknown at the BS. 
The users transmit replicas of their packet according to the random matrix $\mathbf{G} \in \{0,1\}^{T \times M}$, which is called the \emph{access pattern matrix (APM).}
Here, $g_{tm} = [\mathbf{G}]_{tm}$ is the access pattern coefficient of the $m$th user in the $t$th RB.
If $g_{tm} = 1$, the $m$th user transmits its packet in the $t$th RB provided $a_m = 1$, and if $g_{tm} = 0$, the $m$th user does not transmit its packet in the $t$th RB even if $a_m = 1$. 
If $a_m=0$, the $m$th user is \emph{inactive} in the current frame, and does not transmit in any RB. 

At the BS, the received signal in the $t$th RB is a superposition of the packets transmitted by the active users that have chosen to transmit in the $t$th RB. 
The packets of the users undergo both path loss and fading. 
We assume that the path loss component and the second-order statistics of the fading component are known at the BS, and that the fading channel remains constant for the duration of an RB.\footnote{For simplicity of exposition, we consider i.i.d. Rayleigh fading between the users and the BS in each RB, although it is straightforward to extend the results to the correlated fading scenario.} 
Each packet replica comprises a header containing pilot symbols and a payload which includes the coded data and cyclic redundancy check (CRC) symbols. In the pilot phase, if $a_m = 1$, the $m$th user transmits a $\tau -$length pilot sequence $\mathbf{p}_m \in \mathbb{C}^\tau$ in each packet replica (i.e., if $g_{tm}=1$). 
Each pilot symbol has an average power $ P^{\tt{p}}$, and the average power of the pilot sequence is $\mathbb{E} [\|\mathbf{p}_m \|^2] = \tau P^{\tt{p}}$.
The received pilot signal $\mathbf{Y}_t^{\tt{p}} \in  \mathbb{C}^{N \times \tau} $ at the BS across the $N$ antennas in the $t$th RB is thus
\begin{align}
\mathbf{Y}_t^{\tt{p}} &=  \textstyle{\sum\nolimits_{m = 1}^M}  a_m g_{tm} \mathbf{h}_{tm} \mathbf{p}_{m}^{H} + \mathbf{N}_t^{\tt{p}},  \label{eqn_rec_pilot}
\end{align}
where ${\mathbf{N}_t^{\tt{p}}} \in \mathbb{C}^{N \times \tau} $ is the complex additive white Gaussian noise at the BS with $[{\mathbf{N}_t^{\tt{p}}}]_{:,j} \stackrel{\text{i.i.d.}}{\sim} \mathcal{CN}(\mathbf{0}_N,N_0\mathbf{I}_N )$, $\forall \ j \in [\tau]$ and $t \in [T]$, where $N_0$ is the noise variance.
Here, $\mathbf{h}_{tm} = \sqrt{\beta_m} \mathbf{v}_{tm}$ is the uplink channel vector of the $m$th user in the $t$th RB, where $\beta_m$ is the known path loss coefficient and $\mathbf{v}_{tm} $ is the unknown fading vector with $\mathbf{v}_{tm} \stackrel{\text{i.i.d.}}{\sim} \mathcal{CN} (\mathbf{0}_N,\sigma_{\tt{h}}^2\mathbf{I}_N ), \ \forall \ t \in [T]$ and $\ m \in [M]$, and channel variance $\sigma_{\tt{h}}^2$.

In the data phase, if $a_m = 1$, the $m$th user transmits a data symbol\footnote{To derive SINR in any given RB, only one data symbol is written out from the multiple data symbols in each packet.} $x_m$ in each packet replica that it transmits. 
The data symbol satisfies $\mathbb{E}[x_m] = 0$ and $\mathbb{E}[|x_m|^2] = P$, where $P$ denotes the data power. 
The received data signal $\mathbf{y}_t \in \mathbb{C}^N$, at the BS in the $t$th RB, is
\begin{align}
\mathbf{y}_{t}  &=  \textstyle{\sum\nolimits_{m=1}^M} a_m g_{tm} \mathbf{h}_{tm} x_m  + \mathbf{n}_{t}, \label{eqn_rec_data}
\end{align}
where $\mathbf{n}_t \in \mathbb{C}^N$ is the complex additive white Gaussian noise at the BS with $\mathbf{n}_t \stackrel{\text{i.i.d.}}{\sim} \mathcal{CN}(\mathbf{0}_N,N_0\mathbf{I}_N )$, $\forall \ t \in [T]$.

In IRSA, if the $m$th user is active, it samples its repetition factor $d_m$ from a predefined distribution, independently of other users.
Then it chooses $d_m$ RBs from a total of $T$ RBs uniformly at random, and transmits replicas of its packet in these $d_m$ RBs. 
The APM is formed as $[\mathbf{G}]_{tm} = g_{tm}, t \in [T], m \in [M]$, where $g_{tm} = 1$ if the $m$th user has chosen to transmit in the $t$th RB, and $g_{tm} = 0$ otherwise.\footnote{Note that users who are inactive in a given frame can also be virtually considered to have chosen the RBs in which they are scheduled to transmit, even though they do not transmit in any RB.}
This generation of repetition factors is scalable to a large number of users and is completely distributed in nature, and is thus appropriate for MMTC.
In practice, the random subset of RBs is generated using a pseudo-random number generator, whose seed completely determines the sequence~\cite{ref_liva_toc_2011}.\footnote{
For example, the seed could be a function of the current frame index and the user ID, such as, seed $ = [$Current Frame Index; User ID$]$. 
Using simple pseudo-random number generators and with no computational speed-up, we can generate $10^6$ random numbers within a few ms on a mid-range laptop.}
This seed can be pre-programmed at each user, and made available to the BS. 
In this case, it is reasonable to assume that the BS has knowledge of $\mathbf{G}$. 
Also, the APM can be generated in an offline fashion and stored at the BS. 
However, it is important to note that although the BS knows the subset of RBs in which a user would transmit its packet replicas \emph{if it were active,} the BS still does not have the knowledge of which users are active in a given frame.

%\vspace{-0.2cm}
\subsection{SIC-based Decoding in IRSA}\label{sec_sic_dec} 
The decoding process in IRSA proceeds as follows. 
The BS first detects the set of active users (denoted by $\hat{a}_m$). 
Then, it estimates the channels for all the users detected to be active in the RBs for which ${g}_{tm} = 1$. 
It uses these channel estimates to combine the received data signal across the BS antennas, and attempts to decode the user's data packet, treating interference from other users as noise.
If it successfully decodes any user, which can be verified via the CRC, it performs SIC in all RBs which that user has transmitted, with both inter-RB and intra-RB SIC.
The channels are re-estimated for the remaining users, and this decoding process proceeds iteratively.

In this work, the decoding of any user's packet is abstracted into an SINR threshold model as in \cite{ref_clazzer_icc_2017,ref_khaleghi_pimrc_2017}.
That is, the packet can be decoded correctly if and only if the SINR of an active user's packet in an RB is greater than a threshold denoted by $\gamma_{\text{th}}$, the packet can be decoded correctly.
The value of $\gamma_{\text{th}}$  is usually chosen to be $\geq 1$ for a narrowband system \cite{ref_clazzer_icc_2017}; it is a parameter for the purposes of our analysis.

We now briefly describe how to evaluate the performance of IRSA under the abstract SINR threshold-based decoding model. 
We first estimate the user activity coefficients for all users over the frame. 
For the users detected to be active, we compute channel estimates and SINR achieved in all RBs in which the users detected to be active have transmitted their packets.
This SINR accounts for the CSI available at the BS and errors in the UAD process, as we will see in Sec.~\ref{sec_sinr}.
If we find a user with $\text{SINR} \ge \gamma_{\text{th}}$ in some RB, we mark the data packet as having been decoded successfully, and remove that user from the set of users yet to be decoded.
Also, the contribution of the user's packet is removed from all RBs that contain a replica of that packet. 
In the next iteration, the channels are re-estimated from the residual pilot symbols after SIC, the SINRs are recomputed in all RBs, and the decoding of users' packets continues.
The iterations stop when no additional users are decoded in two successive iterations or if all users detected to be active have been successfully decoded. 
The system throughput $\mathcal{T}$ is calculated as the number of correctly decoded unique packets divided by the number of RBs. 
Note that the throughput accounts for packet losses that occur due to users that are incorrectly detected to be active, as well as due to failures in the SIC-based decoding process.

The rest of the paper is organized as follows. Sec. \ref{sec_uad} outlines the proposed UAD algorithm, and Sec. \ref{sec_ch_est} describes the channel estimation process.
The detailed derivation of the SINR accounting for both UAD errors and channel estimation errors is presented later in Sec.~\ref{sec_sinr}.

%\vspace{-0.4cm}
\section{User Activity Detection }\label{sec_uad}
In this section, we describe our user activity detection (UAD) algorithm. 
For this purpose, we consider the conjugate transpose of the received pilot signal in the $t$th RB from \eqref{eqn_rec_pilot} as $\overline{\mathbf{Y}}_t \triangleq {\mathbf{Y}_t^{\tt{p}}}^H $, with $\overline{\mathbf{N}}_t \triangleq {\mathbf{N}_t^{\tt{p}}}^H $.
The signal $\overline{\mathbf{Y}}_t$ can be factorized into the product of the product of two matrices as follows:
\begin{align}
\underbrace{\overline{\mathbf{Y}}_t}_{\tau \times N} 
&= \underbrace{\left[ \mathbf{p}_{1}, \ldots, \mathbf{p}_{M} \right] }_{\mathbf{P}}
\underbrace{\begin{bmatrix}
{a_{1}} g_{t1}  \mathbf{h}_{t1}^{H} \\
\vdots \\
{a_M} g_{tM} \mathbf{h}_{tM}^{H} 
\end{bmatrix}}_{\mathbf{X}_t}
+ \underbrace{\overline{\mathbf{N}}_t}_{\tau \times N}. \label{eqn_uad_ori} % \nonumber \\
%&= \underbrace{\mathbf{P}}_{\tau \times M} \  \underbrace{{\mathbf{X}_t}}_{M \times N} + \underbrace{{\mathbf{N}_t}}_{\tau \times N}.
\end{align}
Here, $\mathbf{P} \in \mathbb{C}^{\tau \times M}$ contains the known pilot sequences of the $M$ users as its columns, and $\mathbf{X}_t \in \mathbb{C}^{M \times N}$ contains the $t$th row of the known APM $\mathbf{G}$, the unknown user activity coefficients, and the unknown channels.
Note that the $i$th row of $\mathbf{X}_t$ is nonzero only if $a_i = 1$ and $g_{ti} = 1$, i.e., when the $i$th user is active and transmits in the $t$th RB.

\begin{table}[t]
\centering
\caption{{Hyperparameter notation in Algorithm~\ref{algo_uad}.}}
\resizebox{0.48\textwidth}{!}{%
\begin{tabular}{|c|l|}
\hline
Symbol & Quantity \\ \hline
$\bm \gamma$  &  Hyperparameter vector of all $M$ users    \\ \hline
$\bm \Gamma$ &    Diagonal matrix with  $\bm \gamma$ as it's diagonal entries     \\ \hline
$\bm \gamma_t$   &  
Hyperparameter vector of the $M_t$ users who would  \\
& have transmitted in the $t$th RB had they been active  \\ \hline
$\bm \Gamma_t$ &    Diagonal matrix with  $\bm \gamma_t$ as it's diagonal entries     \\ \hline
$\bm \gamma^j / \bm \gamma_t^j / \bm \Gamma_t^j $ &    Hyperparameters in the $j$th MSBL iteration     \\ \hline
$\tilde{\bm \gamma}_t^j $ &  Auxiliary variable used to store  $\bm \gamma^j $   \\ \hline
$ \gamma_{\rm pr}$  &  Threshold used to declare support  \\ \hline
\end{tabular}%
}
%\vspace{-5mm}
\label{table_notation}
\end{table}

Let $\mathcal{G}_t = \{ i \in [M] \ | \ g_{ti} = 1 \}$ be the set of users who would have transmitted in the $t$th RB, had they all been active in the current frame, and $M_t = | \mathcal{G}_t |$ be the number of such users.
Since the BS has knowledge of $\mathcal{G}_t$, it can remove the contributions of the users who do not transmit in the $t$th RB.
We thus obtain a column-reduced pilot matrix $\mathbf{P}_t \triangleq [\mathbf{P}]_{:,\mathcal{G}_{t}} \in \mathbb{C}^{\tau \times M_t}$ and a row-reduced channel matrix $\mathbf{Z}_t \triangleq [\mathbf{X}_t]_{\mathcal{G}_{t},:}  \in \mathbb{C}^{M_t \times N}$ in the $t$th RB.
Hence, \eqref{eqn_uad_ori} can be rewritten~as
\begin{align}
\underbrace{\overline{\mathbf{Y}}_t}_{\tau \times N} &= \underbrace{\left[ \mathbf{p}_{{i_1}}, \ldots, \mathbf{p}_{{i_{M_t}}} \right]}_{\mathbf{P}_t} 
\underbrace{\begin{bmatrix}
{a_{{i_1}}}   \mathbf{h}_{t{i_1}}^{H} \\
\vdots \\
{a_{{i_{M_t}}}}  \mathbf{h}_{t{i_{M_t}}}^{H} 
\end{bmatrix}}_{\mathbf{Z}_t}
+ \underbrace{\overline{\mathbf{N}}_t}_{\tau \times N}, \label{eqn_uad_Zt}
%\underbrace{{\mathbf{Y}_t}}_{\tau \times N} & = \underbrace{\mathbf{P}_t}_{\tau \times M_t} \  \underbrace{{\mathbf{X}_t'}}_{M_t \times N} + \underbrace{{\mathbf{N}_t}}_{\tau \times N}.
\end{align}
where $\mathcal{G}_t = \{ i_1, i_2, \ldots, i_{M_t} \}$.
The above results in an under-determined system of linear equations in the MMTC regime (since $\tau \ll M_t \ll M$), with rows of $\mathbf{Z}_t$ being either all zero or all nonzero.
The columns of $\mathbf{Z}_t$ thus share a common support, i.e., they are joint-sparse.
This structure is referred to as a multiple measurement vector (MMV) recovery problem in compressed sensing.
Note that the above step reduces the dimension of the matrix to be estimated, but does not solve the UAD problem.
The support recovery of $\mathbf{Z}_t$ from \eqref{eqn_uad_Zt} can be performed with well known algorithms from compressed sensing literature to recover the activity coefficients in the each of the $T$ RBs.
By doing so, we would obtain an RB-specific activity estimate for each user. 
However, the activity coefficient for any user is the same across the $T$ RBs, and thus we need a way to infer $\{a_i\}$ using information available in all $T$ RBs.
One na\"{i}ve way to do this is to declare users to be active in the current frame if they are detected to be active in at least $\kappa$ RBs, where $\kappa$ is a parameter that can be optimized.
As we will see, this leads to very poor performance compared to the algorithm presented in the sequel. 
In the following paragraph, we briefly discuss Multiple sparse Bayesian learning (MSBL) \cite{ref_wipf_tsp_2007}, which sets the stage for presenting our enhancement that combines the information gleaned from each RB in a principled manner.
The notation we will now use is described in Table~\ref{table_notation}.

%\vspace{-2mm}
MSBL is an empirical Bayesian algorithm that recovers the joint-sparse columns of $\mathbf{Z}_t$ from linear underdetermined measurements $\overline{\mathbf{Y}}_t$.
In MSBL, a hierarchical Gaussian prior is imposed on the columns of $\mathbf{Z}_t$ as
\begin{align}
p(\mathbf{Z}_t; \bm \gamma_t ) = \textstyle{\prod\nolimits_{n=1}^N} p([\mathbf{Z}_t]_{:,n}; \bm \gamma_t) = \textstyle{\prod\nolimits_{n=1}^N} \mathcal{CN}(\mathbf{0}_{M_t}, \bm \Gamma_t), \label{eqn_Zt_dist}
\end{align}
where the columns of $\mathbf{Z}_t$ are i.i.d. and $\bm \Gamma_t =$ diag$(\bm \gamma_t)$.
Here, the elements of $\bm \gamma \in \mathbb{R}_+^M$ are unknown hyperparameters and $\bm \gamma_t \triangleq [\bm \gamma]_{\mathcal{G}_t} \in \mathbb{R}_+^{M_t} $ picks the hyperparameters for the users who would have transmitted in the $t$th RB had they been active in the current frame. 
Such a hierarchical Bayesian model is known to result in sparse solutions for the maximum likelihood estimates of $\bm \gamma_t$~\cite{ref_tipping_2001_jmlr,ref_wipf_tsp_2007}.
Recovering the hyperparameters would yield the users' activities since $[\bm \gamma]_m$ models the variance of the $m$th user's channel.
The hyperparameters in $\bm \gamma_t$ are estimated by maximizing the log-likelihood $\log p(\overline{\mathbf{Y}}_t; \bm \gamma_t)$, with $p(\overline{\mathbf{Y}}_t; \bm \gamma_t ) = \prod_{n=1}^N p([\overline{\mathbf{Y}}_t]_{:,n}; \bm \gamma_t )$.
Here, $p([\overline{\mathbf{Y}}_t]_{:,n}; \bm \gamma_t ) = \mathcal{CN}(\mathbf{0}_{\tau}, \bm \Sigma_{{\bm{\gamma}_t}})$ because of the linear measurement model in \eqref{eqn_uad_Zt}, with $\bm \Sigma_{\bm \gamma_t} = N_0 \mathbf{I}_{\tau} + \mathbf{P}_t \bm \Gamma_t \mathbf{P}_t^H$.
Thus, the log-likelihood reads~as
\begin{align}
\log (p(\overline{\mathbf{Y}}_t; \bm \gamma_t )) \propto -N \log |\bm \Sigma_{\bm \gamma_t} | - \text{Tr} (\bm \Sigma_{\bm \gamma_t}^{-1} \overline{\mathbf{Y}}_t \overline{\mathbf{Y}}_t^H).
\end{align}
This is a non-convex function of $\bm \gamma_t$, and its maximizer cannot be found in closed form.
In MSBL,  expectation maximization (EM) is used to optimize the cost function iteratively~\cite{ref_dempster_jrss_1977}.

Let $j$ denote the iteration index in EM.
In the $j$th MSBL iteration, the notations $\bm \gamma_t^j \triangleq [\bm \gamma^j]_{\mathcal{G}_t}$  and $[\bm \gamma_t^j]_{i}$ denote the hyperparameter vector of the users in the set denoted by $\mathcal{G}_t$ and the $i$th entry of $\bm \gamma_t^j$, respectively.
The EM procedure consists of two steps in each iteration.
The first step, termed the E-step,  updates the covariance $\bm \Sigma_t^{j+1}$ and mean $\bm \mu_{tn}^{j+1}$ of the posterior $p([\mathbf{Z}_t]_{:,n}|[\mathbf{Y}_t]_{:,n}, \bm \gamma_t^j)$ as~\cite{ref_tipping_2001_jmlr} 
\begin{align}
\bm \Sigma_t^{j+1} &= \bm \Gamma_t^j  - \bm \Gamma_t^j  \mathbf{P}_t^H (N_0 \mathbf{I}_{\tau}  + \mathbf{P}_t  \bm \Gamma_t^j   \mathbf{P}_t^H)^{-1}  \mathbf{P}_t \bm \Gamma_t^j, \label{eqn_Estep1}\\
\bm \mu_{tn}^{j+1} &= N_0^{-1} \bm \Sigma_t^{j+1}  \mathbf{P}_t^H  [ \overline{\mathbf{Y}}_t ]_{:,n}, \ n \in [N]. \label{eqn_Estep2}
\end{align}
The second step, known as the M-step, updates the hyperparameter for the $i$th user in the $t$th RB as
\begin{align}
\textstyle{[{ \bm \gamma}_{t}^{j+1}]_i = \dfrac{1}{N} \textstyle{\sum\limits_{n = 1}^N} ( [ \bm \Sigma_t^{j+1}]_{i,i} + |[\bm \mu_{tn}^{j+1}]_i|^2 ), \ i \in [M_t].} \label{eqn_Mstep1}
\end{align}
This M-step estimates the variance of the channel of the $i$th user in the $t$th RB, and this hyperparameter update contains the information from the $t$th RB only.
The above two steps are performed for all $T$ RBs.
Before the next E-step, the hyperparameter updates across the RBs must be combined.

%\vspace{-0.2cm}
\subsection{The Proposed UAD Algorithm}
The main novelty of our UAD algorithm lies in exploiting the access pattern coefficients across RBs to find a single hyperparameter update, which we term as the new M-step.
For this purpose, let $\tilde{ \bm \gamma}_t^{j+1} \in \mathbb{R}_+^M$ be an auxiliary variable for the $t$th RB that is updated as $[\tilde{\bm \gamma}^{j+1}_{t}]_{\mathcal{G}_t} = { \bm \gamma}_{t}^{j+1} \in \mathbb{R}_+^{M_t}$ and zero otherwise.
The hyperparameter update for the $m$th user is obtained at the BS by combining the estimated hyperparameters for that user across all the RBs using the knowledge of $g_{tm}$ as
\begin{align}
[{ \bm \gamma}^{j+1}]_m = \dfrac{1}{d_m} \textstyle{\sum\limits_{t = 1}^T}  g_{tm} [\tilde{\bm \gamma}^{j+1}_{t}]_m, \ m \in [M]. \label{eqn_Mstep2}
\end{align}
Here $d_m = \sum_{t=1}^T g_{tm}$ is the repetition factor of the $m$th user.
Note that, in conjunction with \eqref{eqn_Mstep1}, this new M-step is equivalent to executing an M-step that maximizes the overall log-likelihood, $\sum_{t=1}^T \log p(\overline{\mathbf{Y}}_t; \bm \gamma_t)$, based on the knowledge of the APM at the BS.
Effectively, since it estimates the variance of the channel of the $m$th user by averaging the estimated variances of the channels in each RB, it combines the information obtained from all RBs in computing the hyperparameter update.
By iterating between the E- and M-steps, the EM algorithm converges to a saddle point or a local maximizer of the overall log-likelihood~\cite{ref_dempster_jrss_1977}. 
Further, the EM procedure has been empirically shown to {correctly} recover the support of $\mathbf{Z}_t$, provided $\tau$ and $N$ are large enough~\cite{ref_wipf_tsp_2007}. 
In turn, this leads to significantly lower false positive and false negative rates in UAD, as we will empirically show~later.

\begin{algorithm}[t]
\DontPrintSemicolon
\SetAlgoLined
\SetKwInOut{Input}{Input}
\Input{$ \tau, N, T, M, N_0, \mathbf{G}, \mathbf{P}, \{ \overline{\mathbf{Y}}_t\}_{t=1}^T, \{d_m \}_{m=1}^M, \gamma_{\text{pr}}, j_{\max} $} 
\text{Initialize:} \thinspace $ \bm \gamma^0 = \bm 1_{M} $  \\
\text{Compute:} $ \mathcal{G}_t = \{ m \in [M] \ | \ g_{tm} = 1 \}, \ M_t = |\mathcal{G}_t|, $  $\mathbf{P}_t = [\mathbf{P}]_{:,\mathcal{G}_t} , t \in [T]$ \\
\For{$j = 0,1,2,\ldots, j_{\max}$}{
\For{$t = 1,2,\ldots, T$}{
\text{Compute:} $\bm \Gamma_t^j = \text{diag} ([ \bm \gamma^j]_{\mathcal{G}_t})$ 
\begin{align*}
\bm \Sigma_t^{j+1} &= \bm \Gamma_t^j  - \bm \Gamma_t^j  \mathbf{P}_t^H (N_0 \mathbf{I}_{\tau}  + \mathbf{P}_t  \bm \Gamma_t^j  \mathbf{P}_t^H)^{-1}  \mathbf{P}_t \bm \Gamma_t^j  \\
\bm \mu_{tn}^{j+1} &= N_0^{-1} \bm \Sigma_t^{j+1}  \mathbf{P}_t^H  [ \overline{\mathbf{Y}}_t ]_{:,n}, \ 1 \leq n \leq N 
\end{align*} \\
$ \textstyle{[{ \bm \gamma}_{t}^{j+1}]_i = \dfrac{1}{N} \sum\limits_{n = 1}^N \! ( [ \bm \Sigma_t^{j+1}]_{i,i} + |[\bm \mu_{tn}^{j+1}]_i|^2 ) , \!  i \in [M_t] }$ \\
$[\tilde{\bm \gamma}^{j+1}_{t}]_{\mathcal{G}_t} = { \bm \gamma}_{t}^{j+1}, \ [\tilde{\bm \gamma}^{j+1}_{t}]_{[M] \setminus \mathcal{G}_t}= \bm 0_{M - M_t} $ \\
}
$\textstyle{ [{ \bm \gamma}^{j+1}]_m = \frac{\sum\nolimits_{t = 1}^T g_{tm} [\tilde{\bm \gamma}^{j+1}_{t}]_m  }{\sum\nolimits_{t = 1}^T  g_{tm} } , \ 1 \leq m \leq M }$
} %\vspace{-3pt}
$ \textbf{Output:} \quad \hat{{a}}_m = 
\begin{cases}
1, \  [{ \bm \gamma}^{j_{\max}}]_m \geq \gamma_{\text{pr}} \\
0, \  [{ \bm \gamma}^{j_{\max}}]_m < \gamma_{\text{pr}}
\end{cases} \hspace{-2mm}, \ 1 \leq m \leq M $, 
$ \hat{\mathbf{Z}}_{t} = [ \bm \mu_{t1}^{j_{\max}} \bm \mu_{t2}^{j_{\max}} \ldots \bm \mu_{tN}^{j_{\max}} ], \ 1 \leq t \leq T $ 
\caption{UAD in IRSA}
\label{algo_uad}
\end{algorithm}

The overall UAD procedure is summarized in Algorithm~\ref{algo_uad}.
The algorithm is run for $j_{\max}$  iterations.
As the iterations proceed, the hyperparameters corresponding to  inactive users converge to zero, resulting in sparse estimates. 
At the end of the EM iterations, the estimated activity coefficient $\hat{a}_m$ for the $m$th user is obtained by thresholding $[{ \bm \gamma}^{j_{\max}}]_m$ at a value $\gamma_{\text{pr}}$.
The algorithm also outputs the MAP estimates of the channels $\hat{\mathbf{X}}_{t}$ in each of the $T$ RBs with $[\hat{\mathbf{X}}_{t}]_{\mathcal{G}_t, :} = \hat{\mathbf{Z}}_{t}$ and $[\hat{\mathbf{X}}_{t}]_{[M] \setminus \mathcal{G}_t, :} = \bm 0_{(M - M_t) \times N} $, and the channel estimates of users across all RBs are stacked in $\hat{\mathbf{X}} = [\hat{\mathbf{X}}_{1}, \ldots, \hat{\mathbf{X}}_{T} ]$.

We now discuss the complexity of our algorithm in terms of the number of floating point operations (flops).
Each MSBL iteration has $\mathcal{O}(\tau^2M)$ flops, if the pilot matrix is of size $\tau \times M$~\cite{ref_wipf_tsp_2007}.
In our algorithm, each iteration contains $T$ RBs, where the size of the reduced pilot matrix is $\tau \times M_t$ in the $t$th RB.
Also, the new M-step has lower complexity order than the E-step.
Thus, the overall per-iteration complexity of Algorithm~\ref{algo_uad} is $\mathcal{O}(\tau^2M_S)$, where $M_S = \sum\nolimits_{t=1}^T M_t \approx \bar{d}M$, where $\bar{d}$ is the average repetition factor.

Based on the estimated activity $ \hat{a}_i $ and the true activity $a_i$, the set of all users can be divided into four disjoint subsets
\begin{subequations}
\begin{align} 
\mathcal{A} &= \{i \in [M] \ | \ \hat{a}_i a_i = 1 \}, \\ 
\mathcal{F} &= \{i \in [M] \ | \ \hat{a}_i (1 -a_i) = 1 \}, \\
\mathcal{M} &= \{i \in [M] \ | \ (1 - \hat{a}_i) a_i = 1 \}, \\ 
\mathcal{I} &= \{i \in [M] \ | \ (1 - \hat{a}_i) (1 -a_i) = 1 \}. 
\end{align}
\end{subequations}
$\mathcal{A}$ is the true positive set of users, i.e., the users that are correctly detected to be active. 
$\mathcal{F}$ is the false positive set of users, i.e., the users that are detected to be active and are truly inactive. 
$\mathcal{M}$ is the false negative set of users, i.e., the users that are detected to be inactive, but are actually active. 
$\mathcal{I}$ is the true negative set of users, i.e., the users that are correctly detected to be inactive. 
False positive and false negative users together form the errors in the UAD process, and the error rates for such users are discussed in Sec. \ref{sec_results}.
After the active users are detected, the next task is to estimate the channels from the active users. However, before describing channel estimation, we take a small detour to explain why traditional compressed sensing approaches are not effective for frame-based UAD in IRSA-based multiple access.

%\vspace{-0.2cm}
\subsection{Why One-Shot UAD Does Not Work} \label{sec_one_shot_uad}
By stacking the received signal in \eqref{eqn_uad_ori} across all RBs, we can estimate the user activity coefficients in one-shot across all RBs.
We now briefly explain why this performs poorly.
The received pilot signals in all RBs can be stacked as
\begin{align}
%&\underbrace{\overline{\mathbf{Y}}}_{\tau \times NT} 
&\overline{\mathbf{Y}}
= [\overline{\mathbf{Y}}_1 \overline{\mathbf{Y}}_2 \ldots \overline{\mathbf{Y}}_T]
= \mathbf{PX} + [\overline{\mathbf{N}}_1 \overline{\mathbf{N}}_2 \ldots \overline{\mathbf{N}}_T] , \nonumber  \\
%&\overline{\mathbf{Y}}
%= 
%\begin{bmatrix}
%\overline{\mathbf{Y}}_1 \overline{\mathbf{Y}}_2 \ldots \overline{\mathbf{Y}}_T
%\end{bmatrix} 
%= \mathbf{PX} + {\begin{bmatrix}
%\overline{\mathbf{N}}_1 \overline{\mathbf{N}}_2 \ldots \overline{\mathbf{N}}_T
%\end{bmatrix} } , \nonumber  \\
\mathbf{X}
&= \begin{bmatrix}
{a_{1}} g_{11}  \mathbf{h}_{11}^{H} & \ldots & {a_{1}} g_{T1}  \mathbf{h}_{T1}^{H} \\
\vdots & \ddots & \vdots \\
{a_M} g_{1M} \mathbf{h}_{1M}^{H} & \ldots & {a_M} g_{TM} \mathbf{h}_{TM}^{H} 
\end{bmatrix}  \in \mathbb{C}^{M \times NT}.
\label{eqn_one_shot}
%\underbrace{\left[ \mathbf{p}_{1}, \ldots, \mathbf{p}_{M} \right] }_{\mathbf{P}}
%\underbrace{\begin{bmatrix}
%{a_{1}} g_{11}  \mathbf{h}_{11}^{H} & \ldots & {a_{1}} g_{T1}  \mathbf{h}_{T1}^{H} \\
%\vdots & \ddots & \vdots \\
%{a_M} g_{1M} \mathbf{h}_{1M}^{H} & \ldots & {a_M} g_{TM} \mathbf{h}_{TM}^{H} 
%\end{bmatrix}}_{\mathbf{X}}
%+ \underbrace{{\mathbf{N}}}_{\tau \times NT} . 
\end{align}

The above structure is \emph{not} an MMV recovery problem because the rows of $ \mathbf{X}$ are not completely all zero or all nonzero.
If the $i$th user is inactive, then the $i$th row of $ \mathbf{X}$ is all zero. 
However if the $i$th user is active, then the $i$th row of $ \mathbf{X}$ is not all nonzero.
Only the blocks of the $i$th row corresponding to the RBs in which the $i$th user has transmitted in (i.e., where $g_{ti} = 1$) are all nonzero and the other blocks are all zero.
Since IRSA results in the transmission of replicas in only a small subset of the $T$ RBs, only a few blocks of the $i$th row are nonzero.
Different blocks of each row of $\mathbf{X}$ corresponding to active users have different block-sparse supports.
If an MMV recovery algorithm is applied across all RBs in one shot as in \eqref{eqn_one_shot}, a pilot length of $\tau = \Omega ( M_a \log \frac{M}{M_a} ) $ can achieve a vanishing activity error rate as $N \rightarrow \infty$,  where  $M_a$ is the average number of active users in each RB~\cite{ref_tang_tit_2010}.
For example, with $M=1500$ and $M_a=150$, $\tau = \Omega (346)$ achieves vanishing error rates in a massive MIMO regime.
These pilot lengths are infeasible in practice, and thus, in practical regimes of interest, one-shot UAD performs poorly.

%\vspace{-0.1cm}
\section{Channel Estimation} \label{sec_ch_est}
In addition to performing UAD, Algorithm~\ref{algo_uad} also outputs an initial channel estimate for each user that is detected to be active, as a by-product. 
However, as the decoding iterations proceed, the interference cancellation can help improve the accuracy of the channel estimates, when the channels of the remaining users are re-estimated after each SIC operation.
We now derive MMSE channel estimates in each decoding iteration for all the users that have been detected to be active. 
MMSE channel estimation is also required to compute meaningful expressions for the SINR~\cite{ref_hassibi_tit_2003}.

Since MMSE estimates are recomputed in every iteration, the signals and channel estimates are indexed by the decoding iteration $k$.
Let the set of users who have not yet been decoded in the first $k-1$ iterations be denoted by $\mathcal{S}_{k}$, with $\mathcal{S}_k^m \triangleq \mathcal{S}_k \setminus \{m\}$, and $\mathcal{S}_{1} = [M]$.
The received pilot signal at the BS, in the $t$th RB during the $k$th decoding iteration, is
\begin{align}
{\mathbf{Y}_t^{{\tt{p}}k}} &= \textstyle{\sum\nolimits_{i \in \mathcal{S}_k}}  a_i g_{ti} \mathbf{h}_{ti} \mathbf{p}_{i}^{H} + {\mathbf{N}_t^{\tt{p}}}.
\end{align}
In this section, we assume perfect SIC for simplicity of analysis; we study the performance variation under imperfect SIC in Sec.~\ref{sec_results}.
This received signal is contributed from all users who are truly active in the current frame. 
The BS wishes to compute channel estimates for users who are detected to be active, i.e.,  for the users in $\hat{\mathcal{A}} = \{ i \in [M] \ | \ \hat{a}_i = 1 \}$, which is output by Algorithm \ref{algo_uad}.
For this purpose, the received pilot signal is right combined with the pilot $\mathbf{p}_m, \ \forall m \in \hat{\mathcal{A}} \cap \mathcal{G}_t \cap \mathcal{S}_{k} $, to obtain the post-combining pilot signal as
\begin{align}  
\mathbf{y}_{tm}^{{\tt{p}}k} &= {\mathbf{Y}_t^{{\tt{p}}k}} \mathbf{p}_{m} 
= \textstyle{\sum\nolimits_{i \in \mathcal{S}_k}}  a_i g_{ti} \mathbf{h}_{ti} (\mathbf{p}_{i}^{H} \mathbf{p}_{m} ) + {\mathbf{N}_t^{\tt{p}}} \mathbf{p}_{m}, \label{eqn_rec_postcomb_pilot}
\end{align}
which is further used for estimating the channel between the BS and the user in the $t$th RB \cite{ref_bjornson_mimo_2017}.
The MMSE channel estimate is given by the following theorem.

\begin{thm} \label{thm_ch_est}
The MMSE estimate $ \hat{\mathbf{h}}_{tm}^k$ of the channel $ {\mathbf{h}}_{tm}$ is calculated from the post-combining pilot signal as
\begin{align}
\hat{\mathbf{h}}_{tm}^{k} =  \eta_{tm}^{k} \mathbf{y}_{tm}^{{\tt{p}}k}, \quad \forall m \in \mathcal{S}_k, \label{eqn_ch_est}
\end{align}
where $\eta_{tm}^{k} \triangleq \dfrac{ \hat{a}_m g_{tm} \beta_m \sigma_{\tt{h}}^2 \| \mathbf{p}_m \|^2 }{ N_0 \| \mathbf{p}_m \|^2 + \sum\nolimits_{i \in \mathcal{S}_k} \hat{a}_i g_{ti} \beta_i \sigma_{\tt{h}}^2 | \mathbf{p}_i^H  \mathbf{p}_m|^2 } $.
Further,  the estimation error $ \tilde{\mathbf{h}}_{tm}^k \triangleq \hat{\mathbf{h}}_{tm}^k - {\mathbf{h}}_{tm}$ is uncorrelated with the channel ${\mathbf{h}}_{tm}$, and is distributed as $\mathcal{CN}(\mathbf{0}_N, \delta_{tm}^k \mathbf{I}_N)$.
Here, $\delta_{tm}^{k}$  is the estimation error variance and is given by
\begin{align*}
\delta_{tm}^{k} = \beta_m \sigma_{\tt{h}}^2  \left( \frac{  \sum\nolimits_{i \in \mathcal{S}_k^m} |\mathbf{p}_{i}^{H} \mathbf{p}_{m}|^2 \hat{a}_i a_i g_{ti} \beta_i \sigma_{\tt{h}}^2 + N_0 \|\mathbf{p}_{m}\|^2 }{ \sum\nolimits_{i \in \mathcal{S}_k} |\mathbf{p}_{i}^{H} \mathbf{p}_{m}|^2 \hat{a}_i a_i g_{ti} \beta_i \sigma_{\tt{h}}^2 + N_0 \|\mathbf{p}_{m}\|^2 } \right).
\end{align*}
\end{thm}

\begin{proof}
See Appendix \ref{appendix_ch_est}.
\end{proof}

\noindent \emph{Remark 1}: The channel estimate is composed of a scaling coefficient $\eta_{tm}^{k}$ and the post-combining pilot signal $\mathbf{y}_{tm}^{{\tt{p}}k}$. 
$\eta_{tm}^{k}$ is computed at the BS and is a function of the estimated activity coefficients $\hat{a}_i$. 
Thus, false positive users feature in the denominator of $\eta_{tm}^{k}$ and affect the channel estimates of other users.
The BS also computes channel estimates for these false positive users.\footnote{Since false positive users will fail an error check, the BS can potentially try to identify such users as data decoding proceeds and compute better quality channel estimates.
However, we make no such assumption, and thus, our channel estimation procedure models a worst-case scenario where false positive users contaminate the channel estimates of other true positive users.}
Since false negative users are detected to be inactive, the BS does not account for the interference caused by them while computing $\eta_{tm}^{k}$.
From \eqref{eqn_rec_postcomb_pilot},  $\mathbf{y}_{tm}^{{\tt{p}}k}$ contains signals from other truly active users if pilots are not orthogonal, and is thus a function of the true activity coefficients ${a}_i$. 
Also, false negative users contribute to $\mathbf{y}_{tm}^{{\tt{p}}k}$, and thus both types of errors affect the estimates of other users.

\noindent \emph{Remark 2}: 
In the above theorem, $\delta_{tm}^{k}$ accounts for the pilot contamination from other true positive users.
False positive users are omitted from the expression for $\delta_{tm}^{k}$ because such users do not contaminate the pilots of other users.
Only true positive users contribute to $\delta_{tm}^{k} $.
When orthogonal pilots are used, $\mathbf{p}_i^H \mathbf{p}_m = 0, \forall i \neq m$, there is no pilot contamination, and thus $\delta_{tm}^k = \beta_m \sigma_{\tt{h}}^2 N_0/(\hat{a}_m a_m g_{tm} \beta_m \sigma_{\tt{h}}^2 \| \mathbf{p}_m \|^2 + N_0 )$.

%\vspace{-0.3cm}
\subsection{Cram\'{e}r-Rao Bound} \label{sec_crlb}
In this subsection, we derive the Cram\'{e}r-Rao bound (CRB) \cite{ref_prasad_tsp_2013} on the mean squared error (MSE) of the channel estimated under the hierarchical Bayesian model given by \eqref{eqn_Zt_dist}. 
The signal $\overline{\mathbf{Y}}_t = \mathbf{P}_t \mathbf{Z}_t +  \overline{\mathbf{N}}_t $ from \eqref{eqn_uad_Zt} is first vectorized as 
\begin{align}
\underbrace{\overline{\mathbf{y}}_t}_{N\tau \times 1} &\triangleq \text{vec} (\overline{\mathbf{Y}}_t ) = \underbrace{\bm{\Phi}_t}_{N\tau \times N M_t} \underbrace{\overline{\mathbf{z}}_t}_{N M_t \times 1} +  \underbrace{\overline{\mathbf{n}}_t}_{N\tau \times 1},
\end{align}
where $ \bm{\Phi}_t \triangleq (\mathbf{I}_N \otimes \mathbf{P}_t)$, $\overline{\mathbf{z}}_t \triangleq \text{vec} (\mathbf{Z}_t)$, and $\overline{\mathbf{n}}_t \triangleq \text{vec} (\overline{\mathbf{N}}_t )$.

After stacking the received pilot signal across all RBs as $ \overline{\mathbf{y}} = [\overline{\mathbf{y}}_1^T, \ldots, \overline{\mathbf{y}}_T^T]^T $, with $ \overline{\mathbf{z}} = [\overline{\mathbf{z}}_1^T, \ldots, \overline{\mathbf{z}}_T^T]^T $, $ \overline{\mathbf{n}} = [\overline{\mathbf{n}}_1^T, \ldots, \overline{\mathbf{n}}_T^T]^T $, and $\bm{\Phi} = \text{blkdiag} \{ \bm{\Phi}_1, \ldots, \bm{\Phi}_T \}$, we obtain
\begin{align}
\overline{\mathbf{y}} = \bm{\Phi} \overline{\mathbf{z}} +  \overline{\mathbf{n}}.
\end{align}
Here, we wish to estimate $\overline{\mathbf{z}} \in \mathbb{C}^{N M_S}$ from an observation $\overline{\mathbf{y}} \in \mathbb{C}^{N T \tau}$ via a measurement matrix $\bm{\Phi} \in \mathbb{C}^{N T \tau \times N M_S}$, with $M_S = \sum\nolimits_{t=1}^T M_t$.
Let $\mathbf{J}$ denote the $NM_S \times NM_S$ Fisher information matrix (FIM) associated with the vector $\overline{\mathbf{z}}$.  It is easy to see that $\mathbf{J} = \text{blkdiag} \{ \mathbf{J}_{1}, \ldots, \mathbf{J}_{T} \} $, where $\mathbf{J}_{t}$ is the $NM_t \times NM_t$ sub-block of the FIM corresponding to the $t$th RB.
Specifically,  the CRB derived in this work is the hybrid Cramér-Rao bound~\cite{ref_prasad_tsp_2013},  which is a bound analogous to the CRB for the estimation problem in MSBL. 
Due to the block diagonal structure of the FIM, the CRB for any estimate $\hat{\overline{\mathbf{z}}}_t$ of $\overline{\mathbf{z}}_t$ is given by
\begin{align} 
\mathbb{E}[ ( \hat{\overline{\mathbf{z}}}_t - \overline{\mathbf{z}}_t ) ( \hat{\overline{\mathbf{z}}}_t - \overline{\mathbf{z}}_t )^H] \succeq \mathbf{J}_t^{-1}. \label{eqn_jtinv}
\end{align}

\begin{thm} \label{thm_crlb}
The sub-block of the FIM associated with the channel vector $\overline{\mathbf{z}}_t =$ vec$(\mathbf{Z}_t)$ in the $t$th RB is given by
\begin{align}
\mathbf{J}_t = \mathbf{I}_N \otimes N_0^{-1} \left( \mathbf{P}_t^H \mathbf{P}_t+ {N_0} {\bm{\Gamma}}_t^{-1} \right),
\end{align}
where $ \bm{\Gamma}_t = {\rm diag} ([\bm \gamma]_{\mathcal{G}_t})$ picks the hyperparameters for the $M_t$ users in the $t$th RB.
Further, the CRB for any estimate $[\hat{\mathbf{Z}}_t]_{:,n}$ of $[{\mathbf{Z}}_t]_{:,n}$ in the $t$th RB across the $n$th antenna is given~by
\begin{align}
&\mathbb{E}[ ( [\hat{\mathbf{Z}}_t]_{:,n} - [{\mathbf{Z}}_t]_{:,n} ) ( [\hat{\mathbf{Z}}_t]_{:,n} - [{\mathbf{Z}}_t]_{:,n} )^H] \nonumber \\
& \quad \quad \quad \succeq N_0 \left( \mathbf{P}_t^H \mathbf{P}_t + {N_0}{\bm{\Gamma}}_t^{-1} \right)^{-1}, \quad 1 \leq n \leq N. \label{eqn_crlb}
\end{align}
\end{thm}
\begin{proof}
See Appendix \ref{appendix_crlb}.
\end{proof}
\noindent \emph{Remark}: Note that the right hand side in \eqref{eqn_crlb} is independent of the antenna index. Also, from \eqref{eqn_jtinv}, the MSE of any estimate $\hat{\mathbf{Z}}_t$ of ${\mathbf{Z}}_t$ in the $t$th RB can be bounded below by $\text{Tr} ( \mathbf{J}_t^{-1} )$ as
\begin{align}
&\mathbb{E} [ \| \hat{\mathbf{Z}}_t - \mathbf{Z}_t \|_F^2 ] \geq  \text{Tr} \left( \mathbf{I}_N \otimes {N_0} \left( \mathbf{P}_t^H \mathbf{P}_t + {N_0}  {\bm{\Gamma}}_t^{-1} \right)^{-1} \right) \nonumber \\
& \ = N \text{Tr} \left( {\bm{\Gamma}}_t - {\bm{\Gamma}}_t \mathbf{P}_t^H  \left( N_0\mathbf{I}_\tau +  \mathbf{P}_t {\bm{\Gamma}}_t \mathbf{P}_t^H \right)^{-1} \mathbf{P}_t {\bm{\Gamma}}_t \right), \label{eqn_mse_zt}
\end{align}
where the last step is obtained by using the Woodbury matrix identity and Tr$(\mathbf{I}_N \otimes \mathbf{A}) = N $Tr$(\mathbf{A})$.
Considering the signals received across the entire frame, the effective MSE of the estimate $\hat{\mathbf{X}}$ of $\mathbf{X}$ can thus be bounded as
\begin{align}
\text{MSE} &= \mathbb{E} [ \| \hat{\mathbf{X}} - \mathbf{X} \|_F^2 ] =  \textstyle{\sum\nolimits_{t=1}^T} \mathbb{E} [ \| \hat{\mathbf{X}}_t - \mathbf{X}_t \|_F^2 ], \nonumber \\
&= \textstyle{\sum\nolimits_{t=1}^T} \mathbb{E} [ \| \hat{\mathbf{Z}}_t - \mathbf{Z}_t \|_F^2 ] \nonumber \\
&\geq N N_0 \textstyle{\sum\nolimits_{t=1}^T} \text{Tr} \left( \mathbf{P}_t^H \mathbf{P}_t + {N_0}{\bm{\Gamma}}_t^{-1} \right)^{-1}.
\end{align}
The channel variance can be calculated as
\begin{align}
\mathbb{E} [\| \mathbf{X} \|_F^2] &= \textstyle{\sum\nolimits_{t=1}^T} \mathbb{E} [\| \mathbf{X}_t \|_F^2] = \textstyle{\sum\nolimits_{t=1}^T} \text{Tr} (\mathbb{E} [ \mathbf{Z}_t \mathbf{Z}_t^H ]) \nonumber \\
&= \textstyle{\sum\nolimits_{t=1}^T} \text{Tr} (\mathbf{I}_{N} \otimes {\bm{\Gamma}}_t ) = N \textstyle{\sum\nolimits_{t=1}^T} \text{Tr} ({\bm{\Gamma}}_t ) \nonumber \\
&= N \textstyle{\sum\nolimits_{m=1}^M} d_m[{\bm{\gamma}}]_{m}.
\end{align}
Hence, the normalized mean squared error (NMSE) of any channel estimate $\hat{\mathbf{X}}$ of ${\mathbf{X}}$ can be bounded as
\begin{align}
& \text{NMSE} \triangleq  \frac{\mathbb{E} [\| \mathbf{X} -  \hat{\mathbf{X}} \|_F^2]}{\mathbb{E} [\| \mathbf{X} \|_F^2]} \\
& \ \geq \dfrac{N_0}{\textstyle{\sum\nolimits_{m=1}^M} d_m[{\bm{\gamma}}]_{m}}  \textstyle{\sum\limits_{t=1}^T} \text{Tr} \left( \mathbf{P}_t^H \mathbf{P}_t + {N_0}{\bm{\Gamma}}_t^{-1} \right)^{-1}.\label{eqn_norm_crb}
\end{align}

To better understand the above expressions, we consider the case of orthogonal pilots, i.e., $ \mathbf{P}_t^H \mathbf{P}_t = \tau P^{\tt{p}} \mathbf{I}_{M_t} $,  applicable when $\tau \geq M_t, \ \forall t \in [T]$.
In this case, the MSE is bounded~as
\begin{align*}
& \text{MSE} \geq N \textstyle{\sum\limits_{t=1}^T \sum\limits_{i=1}^{M_t}} \left( \frac{ \tau P^{\tt{p}}}{N_0} + [{\bm{\Gamma}}_t^{-1}]_{i,i} \right)^{-1} \\
& \ = N \textstyle{\sum\limits_{m=1}^{M}} d_m \left( \frac{ \tau P^{\tt{p}}}{N_0} + \frac{1}{[{\bm{\gamma}}]_{m}} \right)^{-1} =  N \textstyle{\sum\limits_{m=1}^{M}} \frac{d_m[{\bm{\gamma}}]_{m} }{ 1 + [{\bm{\gamma}}]_{m} \frac{ \tau P^{\tt{p}}}{N_0} },
\end{align*}
and the NMSE can be bounded as
\begin{align}
\text{NMSE} \geq \dfrac{1}{\sum\nolimits_{m=1}^M d_m[{\bm{\gamma}}]_{m}} \textstyle{\sum\limits_{m=1}^{M}} \frac{d_m[{\bm{\gamma}}]_{m} }{ 1 + [{\bm{\gamma}}]_{m} \frac{ \tau P^{\tt{p}}}{N_0} }.
\end{align}
The above bound is for a given set of repetition factors $\{d_m\}$,  hyperparameters $\bm{\gamma}$, the pilot SNR $\frac{ \tau P^{\tt{p}}}{N_0} $, and is independent of the number of antennas $N$.
As $\tau \rightarrow \infty$, the MSE goes to zero.

We now describe an estimator that achieves the CRB.
\begin{lem} \label{lemma_plug_in}
Assuming the knowledge of the true hyperparameters, the CRB is achieved by the MMSE channel estimate:
\begin{align}
\hat{\mathbf{Z}}_t &= \left(  \mathbf{P}_t^H \mathbf{P}_t  + N_0 {\bm{\Gamma}}_t^{-1} \right)^{-1} \mathbf{P}_t^H \overline{\mathbf{Y}}_t. \label{eqn_plug_in}
\end{align}
\end{lem}
\begin{proof} 
The MSBL algorithm iteratively calculates the MAP estimate.
Since the posterior $p([\mathbf{Z}_t]_{:,n}|[\overline{\mathbf{Y}}_t]_{:,n}; \bm \gamma_t)$ is Gaussian distributed, the MAP estimate is the same as the mean of the distribution, which coincides with the MMSE estimate in \eqref{eqn_plug_in}. 
Upon substituting the above estimate into the MSE expression in \eqref{eqn_mse_zt}, it is easy to show that the CRB is achieved.
\end{proof}
\noindent \emph{Remark:} The above estimator requires knowledge of $ {\bm{\Gamma}}_t$, which in turn needs the user activity coefficients, and is thus a \emph{genie-aided} estimator. 
In practice, one could use the hyperparameter estimates output by Algorithm~\ref{algo_uad} in place of ${\bm{\Gamma}}_t$ to obtain a ``plug-in'' MMSE estimator. 
However, such an estimator need not achieve the CRB. 
Nonetheless, as empirically shown in Sec.~\ref{sec_results}, the channel estimates obtained using \eqref{eqn_Estep2} does achieve the CRB. (See Figs.~\ref{fig_uad_nmse_vs_tau} and \ref{fig_uad_nmse_vs_snr}.)

%\vspace{-0.3cm}
\section{SINR Analysis} \label{sec_sinr}
In this section, the SINR of each user in all the RBs where it has transmitted data is derived, accounting for pilot contamination, estimated user activities, and estimated channels.   
Let $\rho_{tm}^k $ denote the SINR of the $m$th user in the $t$th RB in the $k$th decoding iteration. 
Similar to \eqref{eqn_rec_data}, the received data signal in the $t$th RB and $k$th iteration is
\begin{align}
{\mathbf{y}_t^{k}} &= \textstyle{\sum\nolimits_{i \in \mathcal{S}_k}}  a_i g_{ti} \mathbf{h}_{ti} x_{i} + {\mathbf{n}_t}.
\end{align}
Let $M_t^k = | \hat{\mathcal{A}} \cap \mathcal{G}_{t} \cap  \mathcal{S}_k |$ be the number of users who are detected to be active and have transmitted in the $t$th RB, but have not been decoded in the first $k-1$ iterations. 
$ \hat{\mathcal{A}} $ is obtained as an output of Algorithm \ref{algo_uad}.
A combining matrix $\mathbf{A}_t^{k} \in \mathbb{C}^{N \times M_t^k}$ is used at the receiver in the $t$th RB and $k$th decoding iteration. 
For each $ m \in [M_t^k]$, the vector $\mathbf{a}_{tm}^k = [\mathbf{A}_t^{k}]_{:,m}$ combines the received data signal as
\begin{align}
\tilde{y}_{tm}^k &= [ \mathbf{A}_t^{kH} {\mathbf{y}_{t}^{k}} ]_m =  \mathbf{a}_{tm}^{kH} {\mathbf{y}_{t}^{k}}.
\end{align}

\begin{figure*}[t]
\begin{align}
%\tilde{y}_{tm}^{k} &=  {\mathbf{a}}_{tm}^{kH} \hat{\mathbf{h}}_{tm}^{k}  a_m  g_{tm} x_{m}  - {\mathbf{a}}_{tm}^{kH}  \tilde{\mathbf{h}}_{tm}^{k} a_m g_{tm} x_{m} +  \sum\nolimits_{i \in \mathcal{S}_k^m \cap \mathcal{A}} {\mathbf{a}}_{tm}^{kH} \mathbf{h}_{ti} a_i g_{ti} x_{i} \nonumber \\ 
%& \qquad \qquad +  \sum\nolimits_{i \in \mathcal{S}_k^m \cap \mathcal{M}} {\mathbf{a}}_{tm}^{kH} \mathbf{h}_{ti} a_i g_{ti} x_{i} + {\mathbf{a}}_{tm}^{kH}  \mathbf{n}_{t}. \label{eqn_combining} 
%\tilde{y}_{tm}^{k} &=  {\mathbf{a}}_{tm}^{kH} \hat{\mathbf{h}}_{tm}^{k}  a_m  g_{tm} x_{m}  - {\mathbf{a}}_{tm}^{kH}  \tilde{\mathbf{h}}_{tm}^{k} a_m g_{tm} x_{m} +   \sum\nolimits_{i \in \mathcal{S}_k^m \cap \mathcal{A}} {\mathbf{a}}_{tm}^{kH} \mathbf{h}_{ti} a_i g_{ti} x_{i} +   \sum\nolimits_{i \in \mathcal{S}_k^m \cap \mathcal{M}} {\mathbf{a}}_{tm}^{kH} \mathbf{h}_{ti} a_i g_{ti} x_{i} + {\mathbf{a}}_{tm}^{kH}  \mathbf{n}_{t}.  \label{eqn_combining} 
%\tilde{y}_{tm}^{k} &=  {\mathbf{a}}_{tm}^{kH} \hat{\mathbf{h}}_{tm}^{k}  a_m  g_{tm} x_{m}  - {\mathbf{a}}_{tm}^{kH}  \tilde{\mathbf{h}}_{tm}^{k} a_m g_{tm} x_{m} +   \sum\limits_{i \in \mathcal{S}_k^m \cap \mathcal{A}} {\mathbf{a}}_{tm}^{kH} \mathbf{h}_{ti} a_i g_{ti} x_{i} +   \sum\limits_{i \in \mathcal{S}_k^m \cap \mathcal{M}} {\mathbf{a}}_{tm}^{kH} \mathbf{h}_{ti} a_i g_{ti} x_{i} + {\mathbf{a}}_{tm}^{kH}  \mathbf{n}_{t}.  \label{eqn_combining}
\tilde{y}_{tm}^{k} &=  {\mathbf{a}}_{tm}^{kH} \hat{\mathbf{h}}_{tm}^{k}  a_m  g_{tm} x_{m}  - {\mathbf{a}}_{tm}^{kH}  \tilde{\mathbf{h}}_{tm}^{k} a_m g_{tm} x_{m} +   \textstyle{\sum\nolimits_{i \in \mathcal{S}_k^m \cap \mathcal{A}} } {\mathbf{a}}_{tm}^{kH} \mathbf{h}_{ti} a_i g_{ti} x_{i} +   \textstyle{\sum\nolimits_{i \in \mathcal{S}_k^m \cap \mathcal{M}}} {\mathbf{a}}_{tm}^{kH} \mathbf{h}_{ti} a_i g_{ti} x_{i} + {\mathbf{a}}_{tm}^{kH}  \mathbf{n}_{t}.  \label{eqn_combining} 
\end{align}
%\vspace{-0.8cm}
\end{figure*}

This post-combining data signal is used to decode the $m$th user and is composed of five terms as seen in \eqref{eqn_combining}. 
The term $T_1 \triangleq {\mathbf{a}}_{tm}^{kH} \hat{\mathbf{h}}_{tm}^{k} a_m g_{tm} x_{m}$ is the desired signal of the $m$th user; 
the term $T_2 \triangleq {\mathbf{a}}_{tm}^{kH} \tilde{\mathbf{h}}_{tm}^{k} a_m g_{tm} x_{m}$ is due to the estimation error $\tilde{\mathbf{h}}_{tm}^{k}$ of the $m$th user's channel; 
the term $T_3 \triangleq \sum\nolimits_{i \in \mathcal{S}_k^m \cap \mathcal{A}} {\mathbf{a}}_{tm}^{kH} \mathbf{h}_{ti} a_i g_{ti} x_{i}$ models the inter-user interference from other true positive users (who have transmitted in the $t$th RB and have not yet been decoded); 
the term $T_4 \triangleq \sum\nolimits_{i \in \mathcal{S}_k^m \cap \mathcal{M}} {\mathbf{a}}_{tm}^{kH} \mathbf{h}_{ti} a_i g_{ti} x_{i}$ is the interference from false negative users (who have transmitted in the $t$th RB, but cannot be decoded since they are declared to be inactive); 
and $T_5 \triangleq {\mathbf{a}}_{tm}^{kH} \mathbf{n}_{t}$ is the additive noise. 

To compute the SINR, the power of the post-combining data signal is calculated conditioned on the channel estimates~\cite{ref_bjornson_mimo_2017}. 
This is equivalent to computing the power of the post-combining data signal conditioned on the post-combining pilot signal as $\mathbb{E}_{\mathbf{z}} [|\tilde{y}_{tm}^{k}|^2] = \mathbb{E}_{\mathbf{z}} [|T_1 + T_2 + T_3 + T_4 + T_5|^2] $.  
Here, $\mathbf{z}$ contains the post-combining pilot signals of all $M_t^k$ users yet to be decoded.
Since noise is uncorrelated with data, $T_5$ is uncorrelated with the other terms. 
As MMSE channel estimates are uncorrelated with their estimation errors \cite{ref_bjornson_mimo_2017}, $T_1$ is uncorrelated with $T_2$.
Since the data signals of different users are independent, $T_3$ and $T_4$ are independent of each other and the other terms as well.
Thus, all the five terms are uncorrelated and the power in the received signal is simply the sum of the powers of the individual components
\begin{align} \label{eqn_uncorr}
\mathbb{E}_{\mathbf{z}} [|\tilde{y}_{tm}^{k}|^2] &=  \textstyle{\sum\nolimits_{i=1}^5} \mathbb{E}_{\mathbf{z}} [|T_i|^2].
%\\ &= \mathbb{E}_{\mathbf{z}} [|T_1|^2] + \mathbb{E}_{\mathbf{z}} [|T_2|^2] + \mathbb{E}_{\mathbf{z}} [|T_3|^2]+ \mathbb{E}_{\mathbf{z}} [|T_4|^2]. 
\end{align}
We now compute the SINR in the following theorem.

\begin{thm} \label{thm_sinr}
The signal to interference plus noise ratio (SINR) achieved by the $m$th user in the $t$th RB in the $k$th decoding iteration can be written as
\begin{align}
\rho_{tm}^k &= \dfrac{ {\tt{Gain}}_{tm}^k }{N_0 + {\tt{Est}}_{tm}^k + {\tt{MUI}}_{tm}^k + {\tt{FNU}}_{tm}^k},  \ \forall m \in \mathcal{S}_k, \label{eqn_sinr_all} 
\end{align}
where $ {\tt{Gain}}_{tm}^k $ represents the useful signal power of the $m$th user,  ${\tt{Est}}_{tm}^k$ represents the interference power caused due to estimation errors of all true positive users, $ {\tt{MUI}}_{tm}^k $ represents the multi-user interference power of other true positive users, and ${\tt{FNU}}_{tm}^k$ represents the interference power caused due to the false negative users. 
These can be expressed as
\begin{subequations}
\begin{align}
{\tt{Gain}}_{tm}^k &= P \hat{a}_m a_m g_{tm} \dfrac{ | {\mathbf{a}}_{tm}^{kH} \hat{\mathbf{h}}_{tm}^{k} |^2 }{\| {\mathbf{a}}_{tm}^{k} \|^2} , \\
{\tt{Est}}_{tm}^k &= P \textstyle{\sum\nolimits_{i \in \mathcal{S}_k}} \hat{a}_i a_i g_{ti} \delta_{ti}^{k}, \\
{\tt{MUI}}_{tm}^k  &= P \textstyle{\sum\nolimits_{i \in \mathcal{S}_k^m}} \hat{a}_i a_i g_{ti} \dfrac{ | {\mathbf{a}}_{tm}^{kH} \hat{\mathbf{h}}_{ti}^{k} |^2   }{  \|{\mathbf{a}}_{tm}^{k} \|^2},\\
{\tt{FNU}}_{tm}^k &= P \textstyle{\sum\nolimits_{i \in \mathcal{S}_k^m}} (1- \hat{a}_i) a_i g_{ti} \beta_i \sigma_{\tt{h}}^2.
\end{align}
\end{subequations}
\end{thm}
\begin{proof}
See Appendix \ref{appendix_sinr}.
\end{proof}
\noindent \emph{Remark:} The interference components in the SINR expression are contributed only by truly active users, i.e., the true positive and false negative users.
False positive users do not contribute towards the received data signal.
Even though they do not cause interference, false positive users still affect data decoding of other (true positive) users via their influence on the channel estimates, which also feature in the SINR expression.
Further,  the SINR for such false positive users is zero.\footnote{The BS computes noise-based channel estimates for false positive users. Even if the SINR for such users happens to exceed $\gamma_{\text{th}}$, their packets will fail an error check, and thus their SINR can be set to zero.}
In contrast, false negative users contribute to the received pilot and data signals, affecting both the channel estimates and data decoding of true positive users. 
Since the BS does not detect or decode such users, their SINR is zero as well, and thus the system performance degrades due to such false negative users.
True negative users do not contribute to the received pilot or data signal, and thus do not affect the decoding of other users.
Thus, $\rho_{tm}^k = 0, \ \forall \ m \ \in \mathcal{F} \cup \mathcal{M} \cup \mathcal{I}$.

The SINR expression derived in Theorem \ref{thm_sinr} is applicable to any chosen combining scheme.
For example, with regularized zero forcing combining \cite{ref_bjornson_mimo_2017}, $\mathbf{A}_t^{k} $ is computed as
\begin{align} 
{\mathbf{A}_t^k} =  {\hat{\mathbf{H}}^{k}_t} ({\hat{\mathbf{H}}^{kH}_t} {\hat{\mathbf{H}}^{k}_t} + \lambda \mathbf{I}_{M_t^k}  )^{-1},
\end{align}
where $\lambda$ is the regularization parameter, and ${\hat{\mathbf{H}}_t^k} $ is an $N \times M_t^k$ matrix containing the channel estimates of the $M_t^k$ users as its columns.
The corresponding SINR is obtained by substituting the columns of the above combining matrix into \eqref{eqn_sinr_all}.
The system throughput can now be calculated from \eqref{eqn_sinr_all} via the decoding model described in Sec.~\ref{sec_sic_dec}. 
We note that, in practice, the BS does not compute the SINR; it simply tries to decode each user that is detected to be active, in the RBs it has chosen for transmission. 
However, the decoding succeeds only if the SINR exceeds the chosen threshold. 
Thus, we use the SINR threshold based abstraction to determine which packets are successfully decoded and hence the throughput.

%\vspace{-0.2cm}
\section{Numerical Results} \label{sec_results}
In this section, the UAD and channel estimation performance of Algorithm~\ref{algo_uad} and the impact of UAD errors on the throughput of IRSA are studied via Monte Carlo simulations.
In each run, independent realizations of the user activities, user locations, the APM, and the fades experienced by the users are generated. 
The results in this section are for $T = 50$ RBs, $N_s = 10^3$ Monte Carlo runs, $j_{\max} = 100$ iterations, $\gamma_{\text{pr}}=10^{-4}$, path loss exponent $\alpha = 3.76$, and channel variance $\sigma_{\tt{h}}^2 = 1 $ \cite{ref_bjornson_mimo_2017}.
The pilot sequences are generated as $ \mathbf{p}_m \stackrel{\text{i.i.d.}}{\sim} \mathcal{CN}(\mathbf{0}_{\tau}, P^{\tt{p}} \mathbf{I}_{\tau} )$ as in \cite{ref_liu_tsp_2018}.
The users are spread uniformly at random locations within a cell of radius $r_{\max} = 1000$ m, and the path loss is calculated as $\beta_m = (r_m/r_0)^{-\alpha}$, where $r_m$ is the radial distance of the $m$th user from the BS and $r_0=100$~m is the reference distance.
The soliton distribution \cite{ref_narayanan_istc_2012} with $k_s = 27$ and $a_s = 0.02$ is used to generate the repetition factors.\footnote{The soliton distribution achieves near optimal throughputs~\cite{ref_khaleghi_pimrc_2017}. Here, we reuse the same distribution to generate $d_m$.}

The user activity coefficients are generated as $a_m \stackrel{\text{i.i.d.}}{\sim}\text{Ber}(p_a)$, where $p_a = 0.1$ is the per-user activity probability.
The system load $L$ is defined as the average number of active users per RB, $L = Mp_a/T $. 
The number of users contending for the $T$ RBs is computed in each simulation based on the load $L$ as $M = \lfloor L T/ p_a \rceil$. 
The SNR for the $m$th user is calculated as 
$\text{SNR}_m = P \sigma_{\tt{h}}^2\beta_m/N_0. $
The received SNR of a user at the edge of the cell at the BS is termed as the \emph{cell edge SNR}.
The power levels of all users is set to the same value, $P$, chosen such that the signal from a user at a distance $r_{\max}$ from the BS is received at the cell edge SNR. 
This ensures that all users' signals are received at an SNR greater than or equal to the cell edge SNR, in singleton RBs.\footnote{If the cell edge SNR is such that the cell edge user's packet is decodable, then all users' packets are decodable with high probability in singleton RBs.}
The power levels of users is set to $P = P^{\tt{p}} = 20$ dB~\cite{ref_bjornson_mimo_2017} and $N_0$ is chosen such that the cell edge SNR is $10$~dB, unless otherwise stated.\footnote{In cases where the cell edge SNR is varied, the noise variance $N_0$ is varied according to the required cell edge SNR.}

%\vspace{-0.4cm}
\subsection{Error Rates for UAD}
\begin{figure}[t]
%	\vspace{-0.3cm}
	\centering
\includegraphics[width=0.46\textwidth]{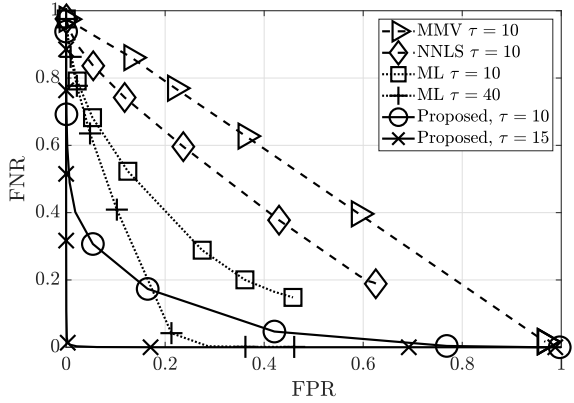}
	\caption{ROC of UAD: comparison with existing approaches.}
	\label{fig_uad_roc}
%	\vspace{-0.2cm}
\end{figure}
\begin{figure}[t]
%	\vspace{-0.3cm}
	\centering
\includegraphics[width=0.46\textwidth]{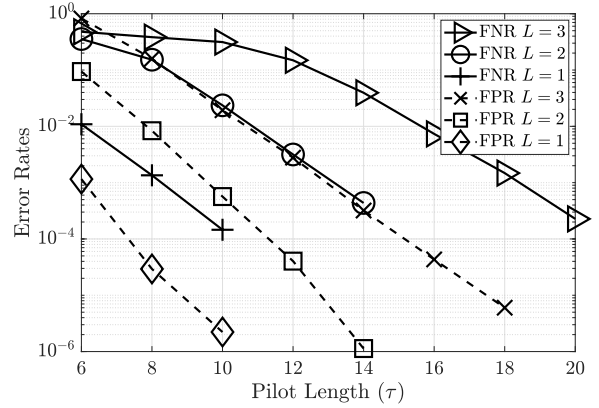}
	\caption{Impact of pilot length on error rates.}
	\label{fig_uad_err_vs_tau}
%	\vspace{-0.5cm}
\end{figure}

In this subsection, the error rates for the recovery of user activity coefficients in IRSA is presented.
The metrics used to characterize the UAD performance are false positive rate, FPR $\triangleq \frac{|\mathcal{F}|}{|\mathcal{F}| + |\mathcal{I}|}$, and false negative rate, FNR $\triangleq \frac{|\mathcal{M}|}{|\mathcal{M}| + |\mathcal{A}|}$. 
FPR is the fraction of inactive users declared to be active whereas FNR is the fraction of active users declared to be inactive.

Fig. \ref{fig_uad_roc} shows the receiver operating characteristic (ROC) plot, and compares the performance of the proposed algorithm with existing approaches such as the maximum likelihood (ML), non-negative least squares (NNLS), and MMV, proposed in \cite{ref_haghighatshoar_isit_2018}.
Here, the threshold $\gamma_{\text{pr}}$ is varied to generate the curves, and the FNR is plotted versus the FPR for $N = 4$ and $L = 3$, which corresponds to $M=1500$ total users.
The existing algorithms are applied to \eqref{eqn_uad_Zt} to detect the $i$th user's activity $\hat{a}_i^{t}$ in the $t$th RB and the user is declared active if it is detected to be active in at least $\kappa$ RB, i.e.,  $\hat{a}_i = \mathbbm{1} \{ \sum_{t=1}^T \hat{a}_i^{t} \geq \kappa \}$.
We use $\kappa=1$ since it yields the lowest FNR.
Note that all of these algorithms estimate users' activities in each RB,  whereas our algorithm combines the estimated hyperparameters in a principled manner as seen in \eqref{eqn_Mstep2}, which is then used to infer the activities, and thus yields far fewer errors.
The proposed algorithm outperforms all three approaches which have themselves shown an improvement over other compressed sensing based algorithms such as approximate message passing \cite{ref_haghighatshoar_isit_2018}.
The ML approach with $\tau = 40$ intersects with the proposed algorithm with $\tau = 10$, and at the point of intersection, Algorithm~\ref{algo_uad} offers a 4-fold reduction in the pilot length compared to the ML approach while achieving the same UAD performance.
Further, the proposed algorithm with $\tau = 15$ significantly outperforms all the other approaches, and achieves a near-ideal performance.

Next, in Fig. \ref{fig_uad_err_vs_tau},  we plot the error rates (i.e., the FNR and FPR) of Algorithm \ref{algo_uad} versus the pilot length for varied $L$ with $N = 16$.
As the load is increased from $L = 1$ to $L=2, 3$, the total number of users over the $T$ RBs increases from $M=500$ to $M=1000, 1500$, and a longer pilot length is needed for accurate UAD.
Thus, there is a significant improvement of the error rates with the pilot length $\tau$.
This is important, since short packets are used in MMTC, and using non-orthogonal pilots with as few as $20$ symbols yields very low error rates with as many as $1500$ users. 
As noted earlier, with classical compressed sensing approaches for UAD, one would require $\Omega(346)$ pilot symbols for accurate UAD in the same settings.

\begin{figure}[t]
%	\vspace{-0.3cm}
	\centering
\includegraphics[width=0.46\textwidth]{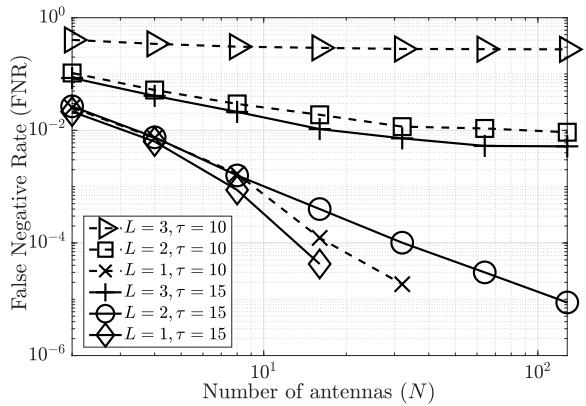}
	\caption{FNR for different pilot lengths and loads.}
	\label{fig_uad_fnr_vs_N}
%	\vspace{-0.3cm}
\end{figure}
\begin{figure}[t]
%	\vspace{-0.2cm}
	\centering
\includegraphics[width=0.46\textwidth]{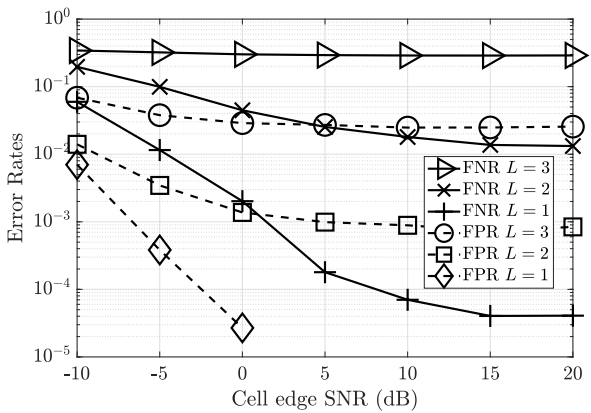}
	\caption{Effect of cell edge SNR on error rates.}
	\label{fig_uad_err_vs_snr}
%	\vspace{-0.4cm}
\end{figure}

Next, we illustrate the variation of the FNR with the number of antennas for varied $L$ and $\tau$, in Fig. \ref{fig_uad_fnr_vs_N}.
The FNR is observed to increase with an increase in $L$.
The FNR also reduces with an increase in $N$ or $\tau$ since the total number of measurements available in the received pilot signal increases, which improves the recovery of user activities in each RB.
For $\tau = 10$, the FNR saturates with $N$ for $L=2/3$, whereas for $\tau = 15$, the FNR saturates at high $L$ and reduces for low to medium $L$.
This is because the performance of MSBL depends more critically on the number of rows $\tau$ in the received signal than the number of columns $N$ \cite{ref_wipf_tsp_2007}.
Thus, at a given load, if $\tau$ is too low,  the FNR improves only slowly with $N$, but if $\tau$ is large enough, the FNR improves dramatically with $N$. Hence, as the load increases, it is important to increase $\tau$ as well. 
In our approach, we solve a reduced problem in each RB as seen in \eqref{eqn_uad_Zt}, after accounting for the APM.
Due to this, in the $t$th RB,  $\tau = \Omega ( M_t p_a \log \frac{M_t}{M_t p_a} ) = \Omega ( - M_t p_a \log p_a)$ would achieve a vanishing error rate.
This guarantee is applicable when $\tau > M_t p_a$, which is the average number of non-zero entries to be recovered in each column of $\mathbf{Z}_t$ in \eqref{eqn_uad_Zt}.
For $k_s = 27$, the average repetition factor is $\bar{d} = 4$, and thus on an average, $M_t = \frac{L \bar{d}}{p_a} = 120, 80, 40$ for $L = 3,2,1$, and of the order $\tau = 28,19,10$ pilot symbols are required, respectively, to achieve a low error rate for Algorithm \ref{algo_uad} as $N$ gets large.

In Fig. \ref{fig_uad_err_vs_snr}, the error rates are plotted against the cell edge SNR for varied $L$ and $\tau=10$.
For low $L$, the error rates first linearly reduce with SNR and then saturate at high SNR.
The FPR for $L=1$ requires longer simulations to capture the point where it saturates with SNR.
For high $L$, both the error rates saturate very quickly with the SNR.
As the load $L$ is decreased, the error rates reduce since there are fewer users to be detected.
As seen earlier, for a fixed $L$, increasing the pilot length can decrease the rates and the error rates reduce at the point of saturation.
In the noise limited regime, i.e., SNR $< -5$ dB, the error rates are high since the Bayesian estimation process performs poorly at such low SNRs.

%\vspace{-0.3cm}
\subsection{Normalized Mean Squared Error}
\begin{figure}[t]
%	\vspace{-0.3cm}
	\centering
\includegraphics[width=0.46\textwidth]{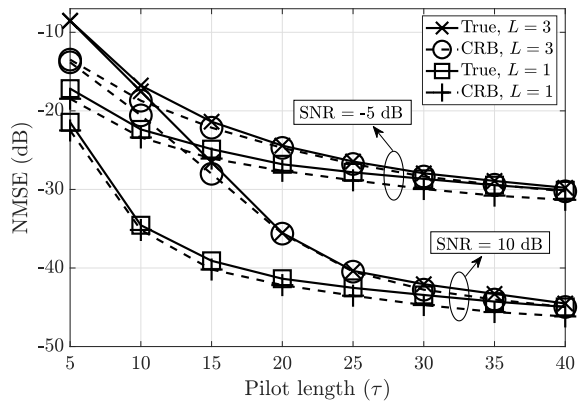}
	\caption{Impact of pilot length on NMSE: The curves labeled \texttt{True} show the NMSE with the estimates output by Algorithm~\ref{algo_uad}.}
	\label{fig_uad_nmse_vs_tau}
%	\vspace{-0.5cm}
\end{figure}
\begin{figure}[t]
%	\vspace{-0.3cm}
	\centering
\includegraphics[width=0.46\textwidth]{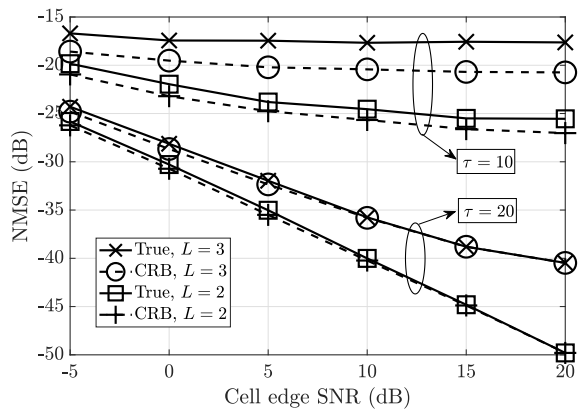}
	\caption{Effect of cell edge SNR on NMSE.}
	\label{fig_uad_nmse_vs_snr}
%	\vspace{-0.5cm}
\end{figure}

Fig. \ref{fig_uad_nmse_vs_tau} shows the impact of the pilot length $\tau$ on the normalized mean squared error (NMSE) of the channels estimated using Algorithm~\ref{algo_uad}.
The NMSE is calculated as  $ \mathbb{E} [\| \mathbf{X} -  \hat{\mathbf{X}} \|_F^2]/\mathbb{E} [\| \mathbf{X} \|_F^2] $, where $\mathbf{X}$ is the channel matrix from \eqref{eqn_one_shot} and $\hat{\mathbf{X}}$ is the corresponding matrix of channel estimates obtained from the UAD algorithm.
It is observed that the NMSE converges to the same value at all $L$ as $\tau$ is increased to $40$, and the value the NMSE converges to decreases with SNR.
As $\tau$ increases,  UAD is perfect and the effect of pilot contamination is reduced, resulting in nearly the same NMSE at all loads.
Also, at low $\tau$, the NMSE is higher for $L=3$ compared to $L=1$, since we have to estimate channels for a larger number of users -- both UAD errors and pilot contamination contribute to a worsening of performance.
The normalized CRB from \eqref{eqn_norm_crb} is also plotted for the system under all the configurations.
It is seen that the gap between the true NMSE and the normalized CRB reduces as $\tau$ increases.
The NMSE is insensitive to the value of $N$, as both the numerator and the denominator of the NMSE scale equally with $N$. 
Hence, we do not study the impact of $N$ on the NMSE.

Fig. \ref{fig_uad_nmse_vs_snr} shows the impact of SNR on NMSE.
The NMSE saturates with an increase in SNR for both $L$ at $\tau=10$, since the UAD performance saturates and any increase in SNR does not improve the quality of the channel estimates.
For $\tau=20$, the NMSE linearly reduces with SNR up to $20$ dB.
At higher $\tau$, the NMSE is lower since there are more measurements available in the received pilot signal to obtain both better UAD performance and high quality channel estimates.
Further, the gap between the true NMSE and the normalized CRB reduces with an increase in SNR for $\tau=20$.
Thus, the CRB, which is achieved by the genie-aided estimator in \eqref{eqn_plug_in}, is also achieved by the estimates in Algorithm \ref{algo_uad} as $\tau$ and SNR are~increased.

%\vspace{-0.3cm}
\subsection{Throughput Accounting for UAD and Channel Estimation}
\begin{figure}[t]
%	\vspace{-0.3cm}
	\centering
\includegraphics[width=0.46\textwidth]{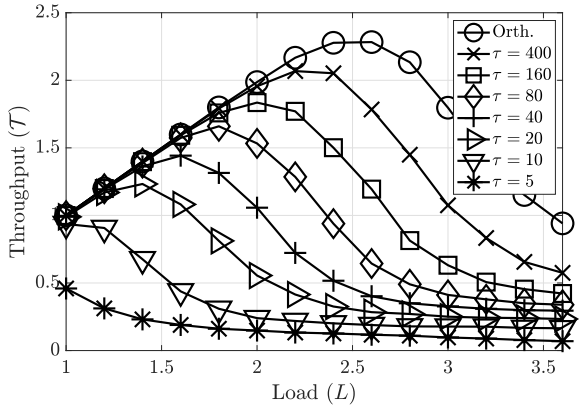}
	\caption{Effect of system load $L$.}
	\label{fig_thpt_vs_L_varying_tau}
%	\vspace{-0.2cm}
\end{figure}
\begin{figure}[t]
%	\vspace{-0.3cm}
	\centering
\includegraphics[width=0.46\textwidth]{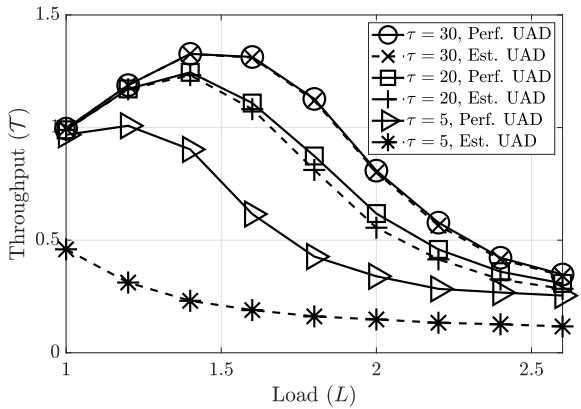}
	\caption{Impact of estimated UAD.}
\label{fig_puad_vs_iuad}
%	\vspace{-0.5cm}
\end{figure}
The performance of IRSA can be characterized by its throughput, which is defined as the number of packets that were successfully decoded at the BS as a fraction of the total number of RBs.\footnote{We note that the throughput $\mathcal{T}$ of IRSA is directly related to the packet loss rate $\textsf{PLR}$ and the spectral efficiency $\mathcal{R}$ as $\mathcal{T} = L(1-\textsf{PLR})$ and $\mathcal{R} = \mathcal{T}\times \log_2(1+\gamma_{\text{th}})$, respectively \cite{ref_liva_toc_2011}. }
Note that, at a system load of $L$, the average throughput of the system is upper bounded by $L$ packets per RB, since there are, on average, $LT$ unique packets  transmitted per frame of duration $T$ RBs.
In this subsection, the SINR analysis presented in Sec.~\ref{sec_sinr} is used to evaluate the throughput of IRSA with UAD and estimated channels.
The number of successfully  decoded packets per RB for each simulation is calculated as described in section \ref{sec_sic_dec}, and the  throughput of the system is found by averaging over the simulations.

%\vspace{-2cm}
Fig. \ref{fig_thpt_vs_L_varying_tau} shows the system throughput, $\mathcal{T}$ (successfully decoded packets per RB), evaluated for different pilot lengths under UAD and estimated CSI, with threshold $\gamma_{\text{th}} = 16$ and regularization parameter $\lambda = 1$, as  a function of the load~$L$.
For $\tau \geq 20$, the throughputs exceed unity, which is the throughput of perfectly coordinated orthogonal access. 
In the moderate load regime ($L<2$), the system can serve more users, and thus the throughput increases linearly with load.
As the load is increased further, the system becomes interference limited as there are too many users sharing the same resources, thereby decreasing the SINR and the throughput.
Also, as the pilot length $\tau$ increases, UAD performance improves, better quality channel estimates are obtained, and the corresponding SINR increases.
The orthogonal pilots curve is obtained by allocating $\tau = M = \lfloor L T/ p_a \rceil$ for each $L$, and this achieves nearly the same performance as the case where perfect CSI is available at the BS.
At $L = 2$, there are $M = 1000$ users that need to be served.
For $\tau = 80$ and $400$, the achievable throughputs are $\mathcal{T} = 1.5$ and $2$, respectively.
At a load of $L = 1.5$, the throughput obtained with $\tau = 80$ is identical to the one offered by the orthogonal pilots, which would need a pilot length of $\tau = M = 750$.
This shows one can use significantly fewer number of pilot symbols and still achieve the same throughput as fully orthogonal pilots, at low to medium loads.

Figure \ref{fig_puad_vs_iuad} quantifies the effect of UAD on the performance of IRSA, by plotting the throughput against the system load under  \emph{perfect} and \emph{estimated} user activities. Here, $\gamma_{\text{th}} = 16$ and $\lambda = 1$ as in the previous figure.
In both cases, the throughput increases linearly with $L$ till it hits a maximum and then reduces.
With a pilot length $\tau = 5$, the gap between estimated and perfect UAD is at its maximum of $0.7$ (packets/RB) at $L = 1.2$.
As the pilot length is increased, the gap reduces to a maximum of $0.1$ (packets/RB) at $L = 2$ for $\tau = 20$ and a negligibly small difference for $\tau = 30$.
This shows that for lower pilot lengths, UAD performance has a significant effect on the throughput.
For higher pilot lengths, the UAD is nearly perfect, and, in this regime, channel estimation and data decoding limits the performance.
Thus, UAD is the easier problem in practical regimes of interest.

\begin{figure}[t]
%	\vspace{-0.3cm}
	\centering
	\includegraphics[width=0.46\textwidth]{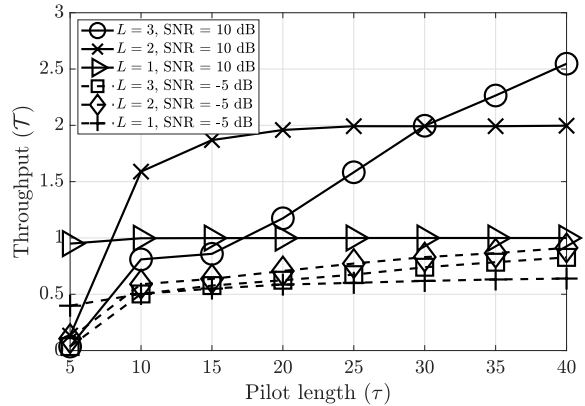}
	\caption{Impact of pilot length $\tau$ on throughput.}
\label{fig_thpt_vs_tau}
%	\vspace{-0.2cm}
\end{figure}
\begin{figure}[t]
%	\vspace{-0.3cm}
	\centering
\includegraphics[width=0.47\textwidth]{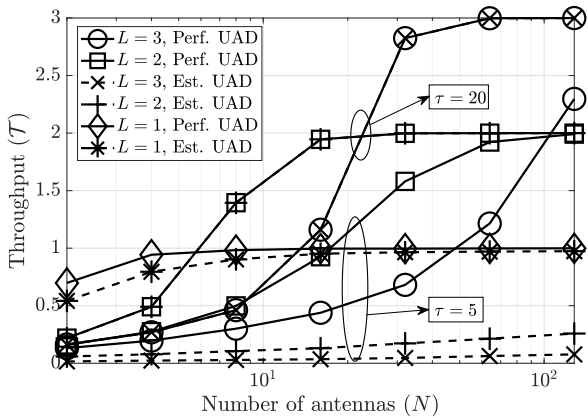}
	\caption{Effect of number of antennas $N$.}
	\label{fig_thpt_vs_N}
%	\vspace{-0.4cm}
\end{figure}

In Fig. \ref{fig_thpt_vs_tau}, we investigate the effect of pilot length on the system throughput at different $L$ and SNRs.
The threshold is set to $\gamma_{\text{th}} = 10$ and regularization parameter to $\lambda = 10^{-2}$ for the rest of the results.
At a cell edge SNR of $-5$ dB, the system throughput is very low due to poor UAD as well as poor quality channel estimates.
The throughput saturates with an increase in $\tau$ for all loads.
Even though more measurements are available at high $\tau$, even if the UAD process is successful and the channel estimates are accurate, the low SNR results in data decoding failures, which limits the throughput.
As the cell edge SNR is increased to $10$ dB, the system performance dramatically improves.
At this SNR, optimal throughputs of $\mathcal{T} = L$ is achieved with $\tau = 10/25$ for $L = 1/2$, respectively, which correspond to $M = 500/1000$ total users and on an average $Mp_a = 50/100$ active users, respectively.
For $L=3$, the optimal throughput is obtained at $\tau=70$, which is not depicted here.
As seen previously, the UAD problem is dominant for very low $\tau$ for these loads, and for higher $\tau$, channel estimation dominates the performance.
To summarize, the pilot length has a significant impact on the performance of IRSA and is instrumental in yielding near-optimal throughputs.

In Fig. \ref{fig_thpt_vs_N}, the system throughput is plotted against the number of antennas at the BS for different $L$ and $\tau$, under both perfect and estimated UAD.
The gap between the perfect and estimated UAD throughputs for $L=2,3$ and $\tau=5$ increases with $N$, and the gap is the highest at $N=128$.
This is because the UAD performance saturates with $N$ for high $L$ at low $\tau$. 
Due to the combined effect of UAD errors, pilot contamination, and interference, low pilot lengths adversely impact both the UAD performance and system throughput.
For $\tau = 20$, increasing $N$ has a dramatic impact at high $L$, and the curves with perfect and estimated UAD overlap completely.
Nearly optimal throughputs of $\mathcal{T} = L$ can be achieved with $N = 16, 32$ antennas for $L = 2,3$.
At $\tau  \ge 20$,  increasing the number of antennas improves UAD, and increases both the array gain and the decoding capability of the BS, leading to more users getting decoded with RZF.
In particular, at $L = 3$, the rise in throughput as $N$ is increased from $8$ to $32$ shows the impact of the number of antennas in improving the throughput.

\begin{figure}[t]
%	\vspace{-0.3cm}
	\centering
\includegraphics[width=0.46\textwidth]{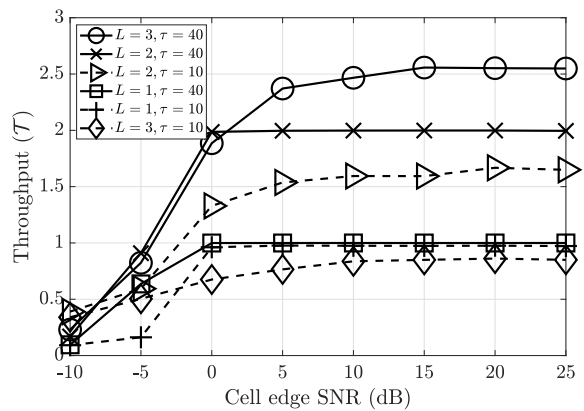}
	\caption{Impact of cell edge SNR.}
	\label{fig_thpt_vs_snr}
%	\vspace{-0.5cm}
\end{figure}
In Fig. \ref{fig_thpt_vs_snr}, we illustrate the impact of cell edge SNR on the throughput.
In the noise-limited regime (SNR $ < 0 $ dB), an increase in SNR increases the throughput. 
Beyond an SNR of $0$ dB, increasing SNR only marginally increases the throughput for all $L$ and $\tau$ and the system becomes interference-limited for $\tau = 10$.
This is because both signal and interference powers get scaled equally, and the SINR remains the same.
At $\tau = 40$, for $L = 1$ and $2$,  optimal throughputs can be obtained at a cell edge SNR $ = 0$ dB. 
However, the throughput for $L = 3$ saturates beyond $10$ to $15$ dB SNR and does not yield the optimal throughput of $\mathcal{T} = 3$ due to high $L$ and low $\tau$.
In summary, the throughput can be improved by increasing the pilot length, number of antennas, and  SNR judiciously: unilaterally increasing one of the three can lead to the throughput saturating at a value lower than $\mathcal{T} = L$.

\begin{figure}[t]
%	\vspace{-0.3cm}
	\centering
\includegraphics[width=0.5\textwidth]{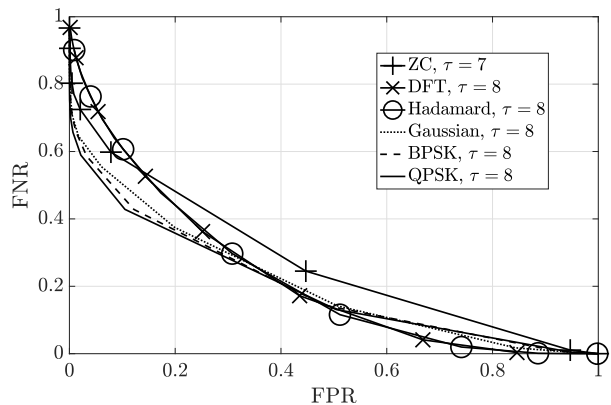}
	\caption{ROC comparison with different pilot sequences.}
	\label{fig_ROC_pilots}
%	\vspace{-0.3cm}
\end{figure}
\begin{figure}[t]
%	\vspace{-0.3cm}
	\centering
\includegraphics[width=0.5\textwidth]{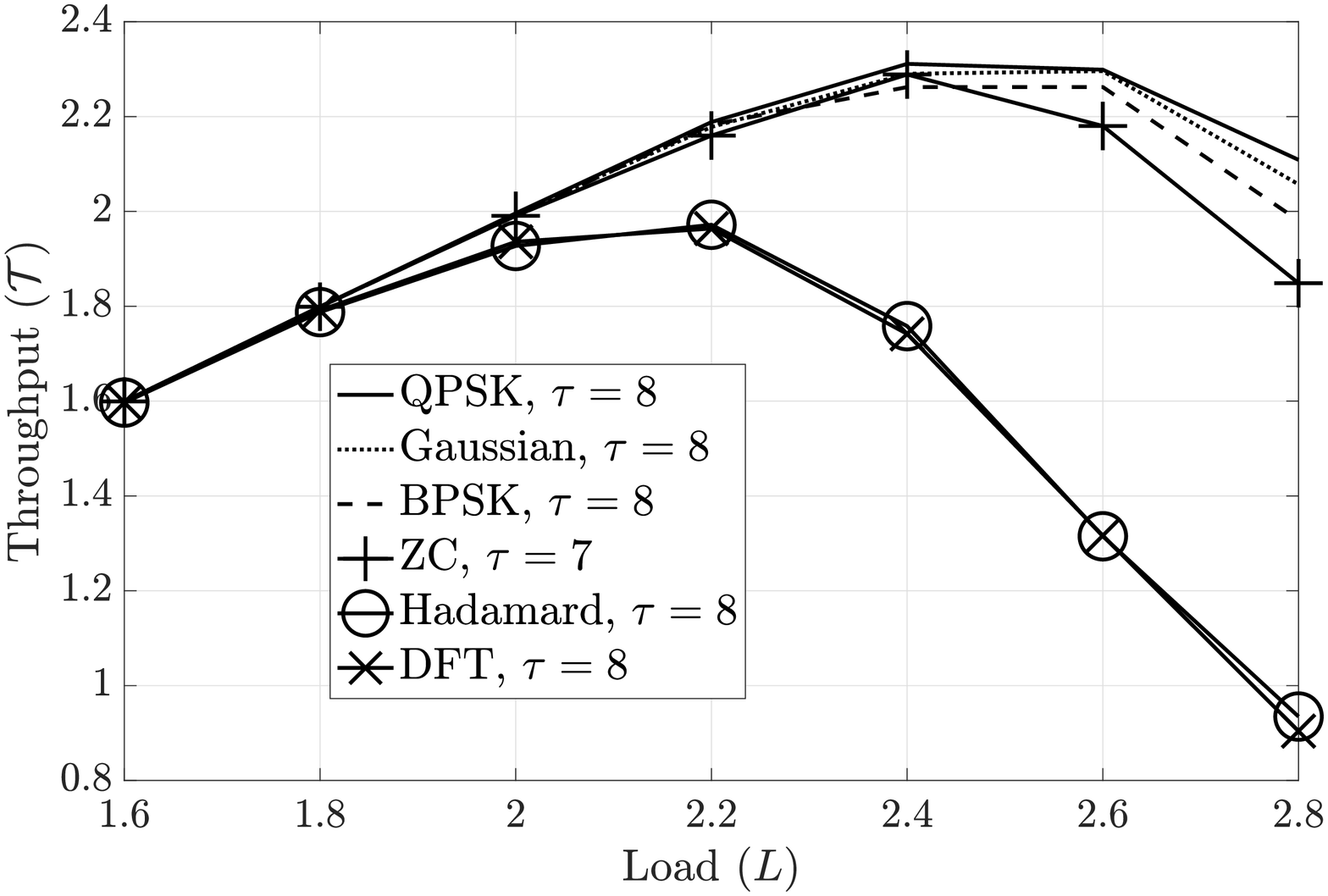}
	\caption{Performance comparison of different pilot sequences.}
	\label{fig_thpt_vs_L_pilots}
%	\vspace{-0.4cm}
\end{figure}

In Fig. \ref{fig_ROC_pilots}, we plot the ROC curves for UAD in IRSA for different pilot sequences, with $N=4$, $L=3$,  and $T=50$.
The non-orthogonal pilots, labeled as BPSK and QPSK, contain random pilot symbols belonging to the respective PSK constellations, and Zadoff-Chu (ZC) sequences are generated according to \cite{ref_zhang_tvt_2012}.
ZC sequences require prime $\tau$;  we use $\tau=7$.
With mutually orthogonal pilot sequences, such as Hadamard and discrete Fourier transform (DFT), $\tau$ sequences of length $\tau$ can be generated.
Thus, we perform orthogonal pilot reuse (OPR), where each user randomly selects a pilot sequence from the available set of $\tau$ pilot sequences, similar to~\cite{ref_valentini_globecom_2021}.
We see that all the pilot sequences have similar UAD performance.
In particular, QPSK, BPSK, and Gaussian pilot sequences have nearly identical performance; DFT and Hadamard sequences have identical UAD performance.

In Fig. \ref{fig_thpt_vs_L_pilots}, we compare the throughput obtained when several orthogonal and non-orthogonal pilot sequence sets are used, with perfect UAD,  $N=16$, $\gamma=6$,  and $\lambda=1$.
Random QPSK, Gaussian, and BPSK pilots have identical performance, and ZC sequences result in a marginally lower throughput at high loads.
In the $\mathcal{T}=L$ regime, the performance of non-orthogonal pilots is better than OPR.
Too much pilot reuse, which is worse with OPR due to the smaller set of available pilots, deteriorates the performance. The use of non-orthogonal pilots provides better diversity, since there is a richer set of pilot sequences, leading to better performance~\cite{ref_senel_tcom_2018}.
Thus, non-orthogonal pilot sequences result in better throughput and nearly identical UAD performance compared to OPR.

\begin{figure}[t]
%	\vspace{-0.3cm}
	\centering
\includegraphics[width=0.5\textwidth]{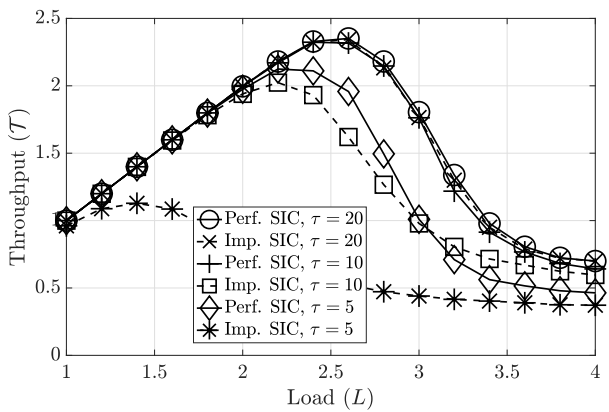}
	\caption{Effect of imperfect SIC.}
	\label{fig_impSIC_thpt_vs_L}
%	\vspace{-0.4cm}
\end{figure}

We now investigate the effect of imperfect SIC due to channel estimation errors on the performance of IRSA.
Under imperfect SIC, the post-combined data signal from \eqref{eqn_combining} contains an extra term, $\sum\nolimits_{i \in \mathcal{S}_1 \setminus \mathcal{S}_k} {\mathbf{a}}_{tm}^{kH} \tilde{\mathbf{h}}_{ti}^{k_i} a_i g_{ti} x_{i}$, that represents the residual interference due to channel estimation, where $k_i$ denotes the iteration in which the $i$th user was decoded.
Thus, the SINR in \eqref{eqn_sinr_all} contains an extra term in the denominator, $ {\tt{ImpSIC}}_{tm}^k = P \sum\nolimits_{i \in \mathcal{S}_1 \setminus \mathcal{S}_k} \hat{a}_i a_i g_{ti} \delta_{ti}^{k_i} $, which represents the power of the residual interference due to imperfect SIC, where $\delta_{ti}^{k_i} $ is the power of the MMSE estimation error of the $i$th user in the $t$th RB who has been decoded in the $k_i$th decoding iteration.
Fig. \ref{fig_impSIC_thpt_vs_L} studies the effect of imperfect SIC on the performance of IRSA, with random BPSK pilots.
We also assume perfect UAD here, since we wish to address the effect of imperfect SIC.
The gap between the perfect SIC and imperfect SIC curves reduce as the pilot length is increased.
The gap is negligible at $\tau=20$, and is very high at $\tau=5$.
Thus, at higher pilot lengths, the effect of imperfect SIC due to channel estimation errors can be ignored.

%\vspace{-0.2cm}
\section{Conclusions} \label{sec_conclusions}
This paper studied the impact of UAD on the throughput of IRSA, which is a GFRA protocol that involves repetition of packets across different RBs.
A novel Bayesian algorithm was proposed to detect the set of active users in IRSA, which exploited the knowledge of the APM, and combined the hyperparameter updates across all RBs to yield an improved UAD performance.
Next, the channel estimates were derived accounting for UAD errors.
A Cram\'{e}r-Rao bound was then derived for the channels estimated under the hierarchical Bayesian model used to develop the proposed algorithm.
Then, the SINR of all users was derived accounting for UAD, channel estimation errors, and pilot contamination.
The effect of these errors on the throughput was studied via extensive simulations.
Many new insights into the design of the IRSA protocol were discussed, namely, the complexity of UAD compared to channel estimation, and the improvement of both UAD and throughput with respect to $\tau$, $N$, SNR, and $L$.
The results underscored the importance of accounting for UAD errors and channel estimation, in studying the throughput offered by the IRSA protocol in MMTC.
We assumed perfect RB- and frame- level synchronization across users and the BS; future work can consider relaxing this assumption.   
Exploiting the asynchronous nature of random access transmissions to detect active users and estimate their channels instead of orthogonal/non-orthogonal pilots is also an interesting direction for future work.

%\appendices
\begin{appendices}
\renewcommand{\thesectiondis}[2]{\Alph{section}:}

%\vspace{-0.8cm}
\section{Proof of Theorem \ref{thm_ch_est}} \label{appendix_ch_est}
Since the channel coefficients are Gaussian distributed, the MMSE estimator is $ \hat{\mathbf{h}}_{tm}^k \triangleq \mathop{{}\mathbb{E}_\mathbf{z}} \left[ {\mathbf{h}}_{tm} \right] $, where $\mathbf{z}$ contains the post-combining pilot signals for all users detected to be active. 
The channel estimation error $\tilde{\mathbf{h}}_{tm} ^{k} \triangleq \hat{\mathbf{h}}_{tm}^{k}  - {\mathbf{h}}_{tm} $ is uncorrelated with the received pilot signal and the estimate itself \cite{ref_bjornson_mimo_2017}. 
The conditional statistics of a Gaussian random vector $\mathbf{x}$ are 
\begin{subequations}
\begin{align}
\mathop{{}\mathbb{E}_\mathbf{z}} \left[{\mathbf{x}} \right]
&= \mathop{{}\mathbb{E}} \left[{\mathbf{x}} \right] + \mathbf{K}_{\mathbf{x} \mathbf{z}} \mathbf{K}_{\mathbf{z} \mathbf{z}}^{-1} \left( \mathbf{z} - \mathop{{}\mathbb{E}} \left[{\mathbf{z}} \right] \right), \label{eqn_gauss_cond1} \\
\mathbf{K}_{\mathbf{xx} | \mathbf{z}}
&= \mathbf{K}_{\mathbf{x} \mathbf{x}} - \mathbf{K}_{\mathbf{x} \mathbf{z}} \mathbf{K}_{\mathbf{z} \mathbf{z}}^{-1} \mathbf{K}_{\mathbf{z} \mathbf{x}}. \label{eqn_gauss_cond2}
\end{align}
\end{subequations}
Here, $\mathbf{K}_{\mathbf{x} \mathbf{x}},$ $\mathbf{K}_{\mathbf{xx} | \mathbf{z}},$ and $ \mathbf{K}_{\mathbf{x} \mathbf{z}}$ are the unconditional covariance of $\mathbf{x}$, the conditional covariance of $\mathbf{x}$ conditioned on $\mathbf{z}$, and the cross-covariance of $\mathbf{x}$ and $\mathbf{z}$, respectively.
From \eqref{eqn_gauss_cond1}, the MMSE channel estimate $ \hat{\mathbf{h}}_{tm}^k$ can be calculated as
\begin{align*}
& \hat{\mathbf{h}}_{tm}^{k} =  \mathop{{}\mathbb{E}} {[ \mathbf{h}_{tm}{\mathbf{y}_{tm}^{{\tt{p}}kH}} ]} \mathop{{}\mathbb{E}} [ \mathbf{y}_{tm}^{{\tt{p}}k} {\mathbf{y}_{tm}^{{\tt{p}}kH}} ]^{-1} \mathbf{y}_{tm}^{{\tt{p}}k} \nonumber \\
&=  \frac{ \hat{a}_m g_{tm} \beta_m \sigma_{\tt{h}}^2 \| \mathbf{p}_m \|^2 }{ (N_0 \| \mathbf{p}_m \|^2 + \sum\nolimits_{i \in \mathcal{S}_k} \hat{a}_i g_{ti} \beta_i \sigma_{\tt{h}}^2 | \mathbf{p}_i^H  \mathbf{p}_m|^2 )}  \mathbf{y}_{tm}^{{\tt{p}}k} \triangleq \eta_{tm}^{k} \mathbf{y}_{tm}^{{\tt{p}}k}.
\end{align*}
The above is computed based on the users detected to be active and is thus a function of estimated activity coefficients $\hat{a}_i$.
From \eqref{eqn_gauss_cond2}, the conditional covariance of the channel~${\mathbf{h}}_{tm}$ is calculated conditioned on $\mathbf{z}$, which contains the post-combining pilot signals for users detected to be active. 
Also, $\mathbf{K}_{{\mathbf{h}}_{tm} {\mathbf{h}}_{tm}}$
$= \beta_m \sigma_{\tt{h}}^2 \mathbf{I}_N$, $\mathbf{K}_{{\mathbf{h}}_{tm}\mathbf{z}}$ $= \mathbb{E} [{\mathbf{h}}_{tm} \mathbf{y}_{tm}^{{\tt{p}}kH}]$ $= \| \mathbf{p}_m \|^2 a_m g_{tm} \beta_m \sigma_{\tt{h}}^2 \mathbf{I}_N$. 
Thus, the conditional covariance is
\begin{align*}
& \mathbf{K}_{{\mathbf{h}}_{tm} {\mathbf{h}}_{tm} | \mathbf{z}}
= \mathbf{K}_{{\mathbf{h}}_{tm} {\mathbf{h}}_{tm}} - \mathbf{K}_{{\mathbf{h}}_{tm}\mathbf{z}} \mathbf{K}_{\mathbf{z} \mathbf{z}}^{-1} \mathbf{K}_{\mathbf{z} {\mathbf{h}}_{tm}} \nonumber \\
& \ = ( \beta_m \sigma_{\tt{h}}^2 - \eta_{tm}^k \| \mathbf{p}_m \|^2 a_m g_{tm} \beta_m \sigma_{\tt{h}}^2 ) \mathbf{I}_N \ \triangleq \ \delta_{tm}^k \mathbf{I}_N. 
%& \delta_{tm}^{k} = \beta_m \sigma_{\tt{h}}^2  \left( \dfrac{  \sum\limits_{i \in \mathcal{S}_k^m} \frac{|\mathbf{p}_{i}^{H} \mathbf{p}_{m}|^2}{\|\mathbf{p}_{m}\|^2} a_i^2 g_{ti}^2 \beta_i \sigma_{\tt{h}}^2 + N_0 }{ \sum\limits_{i \in \mathcal{S}_k} \frac{|\mathbf{p}_{i}^{H} \mathbf{p}_{m}|^2}{\|\mathbf{p}_{m}\|^2} a_i^2 g_{ti}^2 \beta_i \sigma_{\tt{h}}^2 + N_0 } \right) 
\end{align*}
Here, $\delta_{tm}^{k} = \beta_m \sigma_{\tt{h}}^2  \left( \frac{  \sum\nolimits_{i \in \mathcal{S}_k^m} |\mathbf{p}_{i}^{H} \mathbf{p}_{m}|^2 \hat{a}_i a_i g_{ti} \beta_i \sigma_{\tt{h}}^2 + N_0 \|\mathbf{p}_{m}\|^2 }{ \sum\nolimits_{i \in \mathcal{S}_k} |\mathbf{p}_{i}^{H} \mathbf{p}_{m}|^2 \hat{a}_i a_i g_{ti} \beta_i \sigma_{\tt{h}}^2 + N_0 \|\mathbf{p}_{m}\|^2 } \right)$ represents the interference caused due to estimation errors of other true positive users.
It is a function of the pilots of the other true positive users only and not the pilots of false positive users. 
False positive users are omitted from the above because such users do not  contaminate the pilots of other users.
The conditional correlation follows from its definition as
\begin{align*}
\mathbb{E}_{\mathbf{z}}  [{\mathbf{h}}_{tm} {\mathbf{h}}_{tm}^{H}] &= \mathbf{K}_{{\mathbf{h}}_{tm} {\mathbf{h}}_{tm} | \mathbf{z}} +  \mathbb{E}_{\mathbf{z}}  [{\mathbf{h}}_{tm} ] \mathbb{E}_{\mathbf{z}}  [{\mathbf{h}}_{tm}]^H \nonumber \\
&= \delta_{tm}^k \mathbf{I}_N + \hat{\mathbf{h}}_{tm}^{k} \hat{\mathbf{h}}_{tm}^{kH}.
\end{align*} 
The unconditional and conditional means of the error are $\mathbb{E} [\tilde{\mathbf{h}}_{tm}^{k}] = \mathbb{E} [\hat{\mathbf{h}}_{tm}^{k} - {\mathbf{h}}_{tm}] = 0$ and $\mathbb{E}_{\mathbf{z}} [\tilde{\mathbf{h}}_{tm}^{k}] = \mathbb{E}_{\mathbf{z}} [\hat{\mathbf{h}}_{tm}^{k} - {\mathbf{h}}_{tm}] = \hat{\mathbf{h}}_{tm}^{k} - \hat{\mathbf{h}}_{tm}^{k} = 0.$
The conditional covariance of the error~is
\begin{align*}
\mathbf{K}_{\tilde{\mathbf{h}}_{tm}^{k} \tilde{\mathbf{h}}_{tm}^{k} | \mathbf{z}} &= \mathbb{E}_{\mathbf{z}} [\tilde{\mathbf{h}}_{tm}^{k} \tilde{\mathbf{h}}_{tm}^{kH}] = \mathbb{E}_{\mathbf{z}} [(\hat{\mathbf{h}}_{tm}^{k} \! - \! {\mathbf{h}}_{tm}) (\hat{\mathbf{h}}_{tm}^{k} \! - \! {\mathbf{h}}_{tm})^H] \nonumber \\
&=  \mathbb{E}_{\mathbf{z}} [{\mathbf{h}}_{tm} {\mathbf{h}}_{tm}^{H}] - \hat{\mathbf{h}}_{tm}^{k} \hat{\mathbf{h}}_{tm}^{kH} = \delta_{tm}^k \mathbf{I}_N.
\end{align*}
Since $\mathbf{h}_{tm}^k\sim \mathcal{CN}(\mathbf{0}_N, \beta_{m}\sigma_{\tt{h}}^2 \mathbf{I}_N)$,  the  estimate $\hat{\mathbf{h}}_{tm}^k$ and the error $\tilde{\mathbf{h}}_{tm}^k$ are distributed as $\mathcal{CN}$ $(\mathbf{0}_N, $ $\eta_{tm}^k \| \mathbf{p}_m \|^2 a_m g_{tm} \beta_m \sigma_{\tt{h}}^2 \mathbf{I}_N)$ and $\mathcal{CN}(\mathbf{0}_N, \delta_{tm}^k \mathbf{I}_N)$ respectively.
Also, MMSE estimates are uncorrelated with their errors \cite{ref_bjornson_mimo_2017}.

%\vspace{-0.2cm}
\section{Proof of Theorem \ref{thm_crlb}} \label{appendix_crlb}
The FIM sub-block associated with $\overline{\mathbf{z}}_t$ in the $t$th RB is defined as $\mathbf{J}_t = \mathbf{J}_{t1} + \mathbf{J}_{t2}$ \cite{ref_prasad_tsp_2013}, with
\begin{align*}
\mathbf{J}_{t1} &= \mathbb{E} \left[ \left( \dfrac{\partial \log p(\overline{\mathbf{z}}_t)}{\partial \overline{\mathbf{z}}_t^*} \right) \left( \dfrac{\partial \log p(\overline{\mathbf{z}}_t)}{\partial \overline{\mathbf{z}}_t^*} \right)^H \right], \\
\mathbf{J}_{t2} &= \mathbb{E} \left[ \mathbb{E} \left[ \left( \dfrac{\partial \log p(\overline{\mathbf{y}}_t | \overline{\mathbf{z}}_t)}{\partial \overline{\mathbf{z}}_t^*} \right) \left( \dfrac{\partial \log p(\overline{\mathbf{y}}_t| \overline{\mathbf{z}}_t)}{\partial \overline{\mathbf{z}}_t^*} \right)^H \Bigg| \overline{\mathbf{z}}_t \right] \right].
\end{align*}
The conditional probability of $\overline{\mathbf{y}}_t$ given $\overline{\mathbf{z}}_t$ is $\mathcal{CN}(\bm{\Phi}_t\overline{\mathbf{z}}_t, N_0 \mathbf{I}_{\tau N})$, whereas the channel vector $\overline{\mathbf{z}}_t$ is distributed as $\mathcal{CN}(\mathbf{0}_{NM_t}, \mathbf{I}_{N} \otimes {\bm{\Gamma}}_t)$.
Hence, the log of the conditional probabilities behave as
\begin{align*}
\log p(\overline{\mathbf{z}}_t) & \varpropto - \overline{\mathbf{z}}_t^H (\mathbf{I}_{N} \otimes {\bm{\Gamma}}_t)^{-1} \overline{\mathbf{z}}_t, \\
\log p(\overline{\mathbf{y}}_t | \overline{\mathbf{z}}_t) & \varpropto - \dfrac{\| \overline{\mathbf{y}}_t - \bm{\Phi}_t\overline{\mathbf{z}}_t \|_2^2}{N_0}.
\end{align*}
Upon taking the derivative and then calculating the required expectations, it is straightforward to show that $\mathbf{J}_{t1} = \mathbf{I}_N \otimes {\bm{\Gamma}}_t^{-1} $ and $\mathbf{J}_{t2} = \mathbf{I}_N \otimes (\mathbf{P}_t^H \mathbf{P}_t/N_0)$.
Further, the sub-blocks of $\mathbf{J}_t$ corresponding to different antennas are identical and equal to $\mathbf{P}_t^H \mathbf{P}_t/N_0 + {\bm{\Gamma}}_t^{-1}$.
The result follows.

%\vspace{-0.2cm}
\section{Proof of Theorem \ref{thm_sinr}}
\label{appendix_sinr}
In order to compute the SINR, we first compute the power of the individual components.
The desired signal power is 
\begin{align*} 
\mathbb{E}_{\mathbf{z}} [|T_1|^2]  &= \mathbb{E}_{\mathbf{z}} [|{\mathbf{a}}_{tm}^{kH} \hat{\mathbf{h}}_{tm}^{k} a_m g_{tm} x_{m}|^2] =  P a_m^2 g_{tm}^2 |{\mathbf{a}}_{tm}^{kH} \hat{\mathbf{h}}_{tm}^{k}|^2.
\end{align*}
The powers of $a_i$ and $g_{ti} $ are dropped, since they are binary-valued.
In order to account for zero data rates for false positive users, the desired signal power is non-zero only for true positive users and the desired gain is written as
\begin{align} \label{eqn_sig_pow}
{\tt{Gain}}_{tm}^k &\triangleq \frac{ \mathbb{E}_{\mathbf{z}} [|T_1|^2]  }{\| {\mathbf{a}}_{tm}^{k} \|^2} = P \hat{a}_m a_m g_{tm}  \frac{ |{\mathbf{a}}_{tm}^{kH} \hat{\mathbf{h}}_{tm}^{k}|^2 }{\| {\mathbf{a}}_{tm}^{k} \|^2}.
\end{align}
The power of the estimation error term is calculated as
\begin{align*} 
\mathbb{E}_{\mathbf{z}} [|T_2|^2]  &= \mathbb{E}_{\mathbf{z}} [|{\mathbf{a}}_{tm}^{kH} \tilde{\mathbf{h}}_{tm}^{k} a_m g_{tm} x_{m}|^2] = P a_m^2 g_{tm}^2 \delta_{tm}^k \|{\mathbf{a}}_{tm}^{k} \|^2.
\end{align*}
Next, the power of the first inter-user interference term is
\begin{align}
&\mathbb{E}_{\mathbf{z}} [|T_3|^2]  = \mathbb{E}_{\mathbf{z}} \left[ \left| \textstyle{\sum\nolimits_{i \in \mathcal{S}_k^m \cap \mathcal{A}}} {\mathbf{a}}_{tm}^{kH} \mathbf{h}_{ti} a_i g_{ti} x_{i} \right|^2\right] \nonumber \\
& \ \ \ \ = P \textstyle{\sum\nolimits_{i \in \mathcal{S}_k^m \cap \mathcal{A}}} a_i^2 g_{ti}^2 {\mathbf{a}}_{tm}^{kH} \mathbb{E}_{\mathbf{z}}[ \mathbf{h}_{ti} \mathbf{h}_{ti}^H ] {\mathbf{a}}_{tm}^{k} \nonumber \\
& \ \ \ \ \stackrel{(a)}{=}  P \textstyle{\sum\nolimits_{i \in \mathcal{S}_k^m \cap \mathcal{A}}} a_i^2 g_{ti}^2 {\mathbf{a}}_{tm}^{kH} (\delta_{ti}^k \mathbf{I}_N + \hat{\mathbf{h}}_{ti}^{k} \hat{\mathbf{h}}_{ti}^{kH}) {\mathbf{a}}_{tm}^{k} \nonumber \\
& \ \ \ \ =  P \textstyle{\sum\nolimits_{i \in \mathcal{S}_k^m \cap \mathcal{A}}} a_i^2 g_{ti}^2 ( \|{\mathbf{a}}_{tm}^{k} \|^2   \delta_{ti}^k  + | {\mathbf{a}}_{tm}^{kH} \hat{\mathbf{h}}_{ti}^{k}|^2 ),
\end{align}
where $(a)$ follows from Theorem \ref{thm_ch_est}.
Here, $\mathbb{E}_{\mathbf{z}} [|T_2|^2] + \mathbb{E}_{\mathbf{z}} [|T_3|^2]$ represents the contribution of estimation error components of all true positive users and multi-user interference components of other true positive users. 
We now split the normalized version of the above into the sum of the error component ${\tt{Est}}_{tm}^k$ and the multi-user interference ${\tt{MUI}}_{tm}^k$ as follows
\begin{align} 
{\tt{Est}}_{tm}^k &\triangleq P \textstyle{\sum\nolimits_{i \in \mathcal{S}_k}} \hat{a}_i a_i g_{ti} \delta_{ti}^k ,  \label{eqn_est_err_pow} \\
{\tt{MUI}}_{tm}^k  &\triangleq P \textstyle{\sum\nolimits_{i \in \mathcal{S}_k^m}} \hat{a}_i a_i g_{ti} \dfrac{ | {\mathbf{a}}_{tm}^{kH} \hat{\mathbf{h}}_{ti}^{k} |^2   }{  \|{\mathbf{a}}_{tm}^{k} \|^2}  \label{eqn_int_pow_1}.
\end{align}
The power of the second inter-user interference term is
\begin{align} 
&\mathbb{E}_{\mathbf{z}} [|T_4|^2]  = \mathbb{E}_{\mathbf{z}} \left[ \left| \textstyle{\sum\nolimits_{i \in \mathcal{S}_k^m \cap \mathcal{M}}} {\mathbf{a}}_{tm}^{kH} \mathbf{h}_{ti} a_i g_{ti} x_{i} \right|^2\right] \nonumber \\
%& \ \ \ \ = P \textstyle{\sum\nolimits_{i \in \mathcal{S}_k^m \cap \mathcal{M}}} a_i^2 g_{ti}^2 {\mathbf{a}}_{tm}^{kH} \mathbb{E}_{\mathbf{z}}[ \mathbf{h}_{ti} \mathbf{h}_{ti}^H ] {\mathbf{a}}_{tm}^{k} \nonumber \\
& \ \ \ \ \stackrel{(b)}{=} P \textstyle{\sum\nolimits_{i \in \mathcal{S}_k^m \cap \mathcal{M}}} a_i^2 g_{ti}^2 {\mathbf{a}}_{tm}^{kH} \mathbb{E}[ \mathbf{h}_{ti} \mathbf{h}_{ti}^H ] {\mathbf{a}}_{tm}^{k} \nonumber \\
& \ \ \ \ =  P \textstyle{\sum\nolimits_{i \in \mathcal{S}_k^m \cap \mathcal{M}}} a_i^2 g_{ti}^2 {\mathbf{a}}_{tm}^{kH} (\beta_i \sigma_{\tt{h}}^2 \mathbf{I}_N) {\mathbf{a}}_{tm}^{k} \nonumber \\
& \ \ \ \ =  P \textstyle{\sum\nolimits_{i \in \mathcal{S}_k^m \cap \mathcal{M}}} a_i^2 g_{ti}^2 \beta_{i} \sigma_{\tt{h}}^2 \|{\mathbf{a}}_{tm}^{k} \|^2,
\end{align}
\noindent where the conditional expectation is dropped in $(b)$ since the BS does not have the knowledge of the channel estimates of false negative users.
The normalised power of the false negative users is calculated as
\begin{align} \label{eqn_int_pow_2}
{\tt{FNU}}_{tm}^k &\triangleq P \textstyle{\sum\nolimits_{i \in \mathcal{S}_k^m}} (1- \hat{a}_i) a_i g_{ti} \beta_i \sigma_{\tt{h}}^2.
\end{align}
Finally, the noise power is calculated as
\begin{align} \label{eqn_noise_pow}
\mathbb{E}_{\mathbf{z}} [|T_5|^2]  &= \mathbb{E}_{\mathbf{z}} [|{\mathbf{a}}_{tm}^{kH} \mathbf{n}_{t}|^2] =  N_0 \|{\mathbf{a}}_{tm}^{k} \|^2.
\end{align}
Since the five terms in the received signal in \eqref{eqn_combining} are mutually uncorrelated, a meaningful expression for the SINR can be obtained by dividing the useful signal power from \eqref{eqn_sig_pow} by the sum of the interference and the noise powers (which follow from \eqref{eqn_est_err_pow}, \eqref{eqn_int_pow_1}, \eqref{eqn_int_pow_2}, and \eqref{eqn_noise_pow}) \cite{ref_hassibi_tit_2003,ref_bjornson_mimo_2017}.
The SINR can thus be calculated as in \eqref{eqn_sinr_all} for all the users. 
\end{appendices}

\bibliographystyle{IEEEtran}
\bibliography{IEEEabrv,my_refs}

\end{document}